\documentclass[12pt,a4paper]{article}
\usepackage[utf8]{inputenc}
\usepackage[english]{babel}
\usepackage{amsmath}
\usepackage{amsfonts}
\usepackage{amssymb}
\usepackage[margin=1in]{geometry}
\usepackage{amsthm}
\usepackage{hyperref}
\usepackage[numbers]{natbib}

\newtheorem{df}{Definition}[section]
\newtheorem{lem}{Lemma}[section]
\newtheorem{prp}{Proposition}[section]
\newtheorem{thm}{Theorem}[section]
\newtheorem{cor}{Corollary}[section]
\newtheorem{rem}{Remark}[section]
\numberwithin{equation}{section}

\DeclareMathOperator{\diag}{diag}
\DeclareMathOperator{\tr}{tr}

\title{Pricing Interest Rate Derivatives under Volatility Uncertainty}
\author{Julian H\"olzermann\footnote{Center for Mathematical Economics, Bielefeld University, Bielefeld, Germany. Email: julian.hoelzermann@uni-bielefeld.de. The author thanks Frank Riedel for valuable advice and Max Nendel, Wolfgang Runggaldier, and the participants of the ``13th European Summer School in Financial Mathematics'' in Vienna for fruitful discussions. The author gratefully acknowledges financial support by the German Research Foundation (Deutsche Forschungsgemeinschaft) via Collaborative Research Center 1283.}}

\begin{document}

\maketitle

\begin{abstract}
\noindent In this paper, we study the pricing of contracts in fixed income markets under volatility uncertainty in the sense of Knightian uncertainty or model uncertainty. The starting point is an arbitrage-free bond market under volatility uncertainty. The uncertainty about the volatility is modeled by a $G$-Brownian motion, which drives the forward rate dynamics. The absence of arbitrage is ensured by a drift condition. Such a setting leads to a sublinear pricing measure for additional contracts, which yields either a single price or a range of prices. Similar to the forward measure approach, we define the forward sublinear expectation to simplify the pricing of cashflows. Under the forward sublinear expectation, we obtain a robust version of the expectations hypothesis, and we show how to price options on forward prices. In addition, we develop pricing methods for contracts consisting of a stream of cashflows, since the nonlinearity of the pricing measure implies that we cannot price a stream of cashflows by pricing each cashflow separately. With these tools, we derive robust pricing formulas for all major interest rate derivatives. The pricing formulas provide a link to the pricing formulas of traditional models without volatility uncertainty and show that volatility uncertainty naturally leads to unspanned stochastic volatility.
\end{abstract}


\section{Introduction}
The present paper deals with the pricing of interest rate derivatives under volatility uncertainty in the sense of Knightian uncertainty or model uncertainty, also referred to as ambiguous volatility. Due to the assumption of a single, known probability measure, traditional models in finance are subject to \textit{model uncertainty}---that is, the uncertainty about using the correct probability measure---since it is not always possible to specify the probabilistic law of the underlying. Therefore, a new stream of research, called \textit{robust finance}, emerged in the literature, examining financial markets in the presence of a family of probability measures (or none at all) to obtain a robust model. The most frequently studied type of model uncertainty is volatility uncertainty: the volatility determines the probabilistic law of the underlying, but there are many ways to model the volatility of an underlying and it is unknown which describes the future evolution of the volatility best. The literature on robust finance has led to pricing rules that are robust with respect to the volatility. The aim of this paper is to develop robust pricing rules for contracts traded in fixed income markets.
\par The initial setting is an arbitrage-free bond market under volatility uncertainty. The uncertainty about the volatility is represented by a family of probability measures, called \textit{set of beliefs}, consisting of all beliefs about the volatility. This framework naturally leads to a sublinear expectation and a $G$-Brownian motion. A $G$-Brownian motion, which was invented by \citet*{peng2019}, is basically a standard Brownian motion with an ambiguous volatility---the volatility is completely uncertain but bounded by two extremes. We model the bond market in the spirit of \citet*{heathjarrowmorton1992} (HJM); that is, we model the instantaneous forward rate as a diffusion process, which is driven by a $G$-Brownian motion. The remaining quantities on the bond market are defined in terms of the forward rate in accordance with the HJM methodology. We model the forward rate in such a way that it satisfies a suitable drift condition, ensuring the absence of arbitrage on the bond market. Additionally, we assume that the diffusion coefficient of the forward rate is deterministic, which enables us to derive pricing methods for typical derivatives and corresponds to an HJM model in which the foward rate is normally distributed.
\par In the presence of volatility uncertainty, we obtain a sublinear pricing measure for additional contracts we add to the bond market, which yields either a single price or a range of prices. Within the framework described above, we consider additional contracts, which we want to price without admitting arbitrage. The pricing of contracts under volatility uncertainty is different from the classical approach, since the expectation---which corresponds to the pricing measure in the classical case without volatility uncertainty---is sublinear in this setting. In contrast to the classical case, we use the sublinear expectation to determine the price of a contract or its bounds; hence, we refer to it as the \textit{risk-neutral sublinear expectation}. To show that this approach indeed yields arbitrage-free prices, we define trading strategies and arbitrage on the bond market extended by the additional contract. Then we show that the extended bond market is arbitrage-free, meaning that we can use this approach to find no-arbitrage prices for contracts.
\par To simplify the pricing of single cashflows, we introduce a counterpart of the forward measure, called \textit{forward sublinear expectation}. The forward measure, invented by \citet*{geman1989}, is used for pricing discounted cashflows in classical models without volatility uncertainty \citep*{bracemusiela1994,gemanelkarouirochet1995,jamshidian1989}. We define the forward sublinear expectation by a $G$-backward stochastic differential equation (BSDE) and show that it corresponds to the expectation under the forward measure. Similar to the forward measure, the forward sublinear expectation has the advantage that computing the sublinear expectation of discounted cashflows reduces to computing the forward sublinear expectation of cashflows, discounted with the bond price. Under the forward sublinear expectation, we obtain several results needed for pricing cashflows of typical fixed income products. As a by-product, we obtain a robust version of the expectations hypothesis under the forward sublinear expectation. Moreover, we provide pricing methods for options on forward prices. The prices of such options are characterized by nonlinear partial differential equations (PDEs) or, in some cases, by the prices from the corresponding HJM model without volatility uncertainty.
\par In addition, we develop pricing methods for contracts consisting of several cashflows. In traditional models without volatility uncertainty, there is no distinction between pricing single cashflows and pricing a stream of cashflows, since the pricing measure is linear. However, when there is uncertainty about the volatility, the nonlinearity of the pricing measure implies that we cannot generally price a stream of cashflows by pricing each cashflow separately. Therefore, we provide different schemes for pricing a family of cashflows. If the cashflows of a contract are sufficiently simple, we can price the contract as in the classical case. In general, we use a backward induction procedure to find the price of a contract. When the contract consists of a family of options on forward prices, the price of the contract is characterized by a system of nonlinear PDEs or, in some cases, by the price from the corresponding HJM model without volatility uncertainty.
\par With the tools mentioned above, we derive robust pricing formulas for all major interest rate derivatives. We consider typical linear contracts, such as fixed coupon bonds, floating rate notes, and interest rate swaps, and nonlinear contracts, such as swaptions, caps and floors, and in-arrears contracts. Due to the linearity of the payoff, we obtain a single price for fixed coupon bonds, floating rate notes, and interest rate swaps; the pricing formula is the same as the one from classical models without volatility uncertainty. Due to the nonlinearity of the payoff, we obtain a range of prices for swaptions, caps and floors, and in-arrears contracts; the range is bounded from above, respectively below, by the price from the corresponding HJM model without volatility uncertainty with the highest, respectively lowest, possible volatility. Therefore, the pricing of common interest rate derivatives under volatility uncertainty reduces to computing prices in models without volatility uncertainty. For other (less common) contracts the pricing procedure requires (novel) numerical methods.
\par The pricing formulas show that volatility uncertainty is able to naturally explain empirical findings that many traditional term structure models fail to reproduce. According to empirical evidence, volatility risk in fixed income markets cannot be hedged by trading solely bonds, which is termed \textit{unspanned stochastic volatility} and inconsistent with traditional term structure models \citep*{collindufresnegoldstein2002}. Since the presence of volatility uncertainty naturally leads to market incompleteness, the pricing formulas derived in this paper show that it is no longer possible to hedge volatility risk in fixed income markets with a portfolio consisting solely of bonds when there is uncertainty about the volatility. Moreover, the pricing formulas are in line with the empirical findings of \citet*{collindufresnegoldstein2002}.
\par Apart from giving a natural explanation for empirical findings, the theoretical results can be used in practice for different purposes. One can use the pricing procedure for stress testing by pricing contracts in the presence of different levels of volatility uncertainty and investigating how the pricing bounds behave compared to the price from the corresponding HJM model without volatility. One can also fit the pricing bounds to bid-ask spreads of quoted prices to obtain the bounds for the volatility and use them to price other contracts. Alternatively, the bounds for the volatility can be inferred from historical data on the volatility in the form of confidence intervals to generally price contracts.
\par The literature on model uncertainty and, especially, volatility uncertainty in financial markets or, primarily, asset markets is very extensive. The first to apply the concept of volatility uncertainty to asset markets were \citet*{avellanedalevyparas1995} and \citet*{lyons1995}. Over a decade afterwards, the topic gained a lot of interest \citep*{epsteinji2013,vorbrink2014}. The interesting fact about volatility uncertainty is that it is represented by a nondominated set of probability measures. Hence, traditional results from mathematical finance like the fundamental theorem of asset pricing break down. There are various attempts to extend the theorem to a multiprior setting \citep*{bayraktarzhou2017,biaginibouchardkardarasnutz2017,bouchardnutz2015}. In some situations the theorem can be even extended to a model-free setting, that is, without any reference measure at all \citep*{acciaiobeiglbockpenknerschachermayer2016,burzonifrittellihoumaggisobloj2019,riedel2015}. Most of those works also deal with the problem of pricing and hedging derivatives in the presence of model uncertainty. The topic has been studied separately in the presence of volatility uncertainty \citep*{vorbrink2014}, in the presence of a general set of priors \citep*{aksamitdengoblojtan2019,carassusoblojwiesel2019}, and in a model-free setting \citep*{bartlkupperpromeltangpi2019,beiglbockcoxhuesmannperkowskipromel2017}. The most similar setting is the one of \citet*{vorbrink2014}, since it focuses on volatility uncertainty modeled by a $G$-Brownian motion. However, the focus, as in most of the studies from above, lies on asset markets.
\par In addition, there is an increasing number of articles dealing with interest rate models or related credit risk under model uncertainty \citep*{acciaiobeiglbockpammer2021,avellanedalewicki1996,biaginioberpriller2021,biaginizhang2019,epsteinwilmott1999,fadinaneufeldschmidt2019,fadinaschmidt2019,holzermann2021,holzermann2021'}. Among those, there are also articles focusing on volatility uncertainty in interest rate models \citep*{avellanedalewicki1996,fadinaneufeldschmidt2019,holzermann2021,holzermann2021'}. The only one working in a general HJM framework is a companion paper \citep*{holzermann2021'}. The remaining articles on model uncertainty in interest rate models either correspond to short rate models or do not study volatility uncertainty. The main result of the accompanying article \citep*{holzermann2021'} is a drift condition, which shows how to obtain an arbitrage-free term structure in the presence of volatility uncertainty. Starting from an arbitrage-free term structure, the aim of the present paper is to study the pricing of derivatives in fixed income markets under volatility uncertainty.
\par There are several ways to describe volatility uncertainty from a mathematical point of view. The classical approach is the one of \citet*{denismartini2006} and \citet*{peng2019}. Actually, those are two different approaches, but they are equivalent as it was shown by \citet*{denishupeng2011}. The difference is that \citet*{denismartini2006} start from a probabilistic setting, whereas the calculus of $G$-Brownian motion from \citet*{peng2019} relies on nonlinear PDEs. Moreover, there are various extensions and generalizations \citep*{nutz2013,nutzvanhandel2013}. Additional results and a different approach to volatility uncertainty were developed by \citet*{sonertouzizhang2011a,sonertouzizhang2011b}. There are also many attempts to a pathwise stochastic calculus, which works without any reference measure \citep*[and references therein]{contperkowski2019}. In this paper, we use the calculus of $G$-Brownian motion, since the literature on $G$-Brownian motion contains a lot of results. In particular, the results of \citet*{hujipengsong2014} are of fundamental importance for the results derived in this paper.
\par The remainder of this paper is organized as follows. Section \ref{arbitrage-free bond market} introduces the overall setting of the model: an arbitrage-free bond market under volatility uncertainty. In Section \ref{risk-neutral valuation}, we show that we can use the risk-neutral sublinear expectation as a pricing measure for additional contracts. In Section \ref{pricing single cashflows}, we define the forward sublinear expectation and derive related results for the pricing of single cashflows. Section \ref{pricing a stream of cashflows} provides schemes for pricing contracts consisting of a stream of cashflows. In Section \ref{common interest rate derivatives}, we derive pricing formulas for the most common interest rate derivatives. In Section \ref{market incompleteness}, we discuss market incompleteness and show that volatility uncertainty leads to unspanned stochastic volatility. Section \ref{conclusion} gives a conclusion.

\section{Arbitrage-Free Bond Market}\label{arbitrage-free bond market}
We represent the (Knightian) uncertainty about the volatility by a familiy of probability measures such that each measure corresponds to a specific belief about the volatility. Let us consider a probability space $(\Omega,\mathcal{F},P_0)$ such that the canonical process $B=(B_t^1,...,B_t^d)_{t\geq0}$ is a $d$-dimensional standard Brownian motion under $P_0$. Furthermore, let $\mathbb{F}=(\mathcal{F}_t)_{t\geq0}$ be the filtration generated by $B$ and completed by all $P_0$-null sets. The state space of the uncertain volatility is given by
\begin{align*}
\Sigma:=\big\{\sigma\in\mathbb{R}^{d\times d}\,\big\vert\,\sigma=\diag(\sigma_1,...,\sigma_d),\,\sigma_i\in[\underline{\sigma}_i,\overline{\sigma}_i]\,\text{for all}\,i=1,...,d\big\},
\end{align*}
where $\overline{\sigma}_i\geq\underline{\sigma}_i>0$ for all $i$; that means, we consider all scenarios in which there is no correlation and the volatility is bounded by two extremes: the matrices $\overline{\sigma}:=\diag(\overline{\sigma}_1,...,\overline{\sigma}_d)$ and $\underline{\sigma}:=\diag(\underline{\sigma}_1,...,\underline{\sigma}_d)$. For each $\Sigma$-valued, $\mathbb{F}$-adapted process $\sigma=(\sigma_t)_{t\geq0}$, we define the process $B^\sigma=(B_t^\sigma)_{t\geq0}$ by
\begin{align*}
B_t^\sigma:=\int_0^t\sigma_udB_u
\end{align*}
and the measure $P^\sigma$ to be the probabilistic law of the process $B^\sigma$, that is,
\begin{align*}
P^\sigma:=P_0\circ(B^\sigma)^{-1}.
\end{align*}
We denote the collection of all such measures by $\mathcal{P}$, which is termed the \textit{set of beliefs}, since it contains all beliefs about the volatility. Now the canonical process has a different volatility under each measure in the set of beliefs.
\par Volatility uncertainty naturally leads to a $G$-expectation and a $G$-Brownian motion. If we define the sublinear expectation $\hat{\mathbb{E}}$ by
\begin{align*}
\hat{\mathbb{E}}[\xi]:=\sup_{P\in\mathcal{P}}\mathbb{E}_P[\xi]
\end{align*}
for all random variables $\xi$ such that $\mathbb{E}_P[\xi]$ exists for all $P\in\mathcal{P}$, then $\hat{\mathbb{E}}$ corresponds to the $G$-expectation on $L_G^1(\Omega)$ and $B$ is a $G$-Brownian motion under $\hat{\mathbb{E}}$ \citep*[Theorem 54]{denishupeng2011}. The letter $G$ refers to the sublinear function $G:\mathbb{S}^d\rightarrow\mathbb{R}$, given by
\begin{align*}
G(A):=\tfrac{1}{2}\sup_{\sigma\in\Sigma}\tr(\sigma\sigma'A)=\tfrac{1}{2}\sum_{i=1}^d\big(\overline{\sigma}_i^2(a_{ii})^+-\underline{\sigma}_i^2(a_{ii})^-\big),
\end{align*}
where $\mathbb{S}^d$ is the space of all symmetric $d\times d$ matrices and $\cdot'$ denotes the transpose of a matrix. The function $G$ is the generator of the nonlinear PDE that defines the $G$-expectation and characterizes the distribution and the uncertainty of a $G$-Brownian motion. The space $L_G^1(\Omega)$ is the space of random variables for which the $G$-expectation is defined. We identify random variables in $L_G^1(\Omega)$ if they are equal \textit{quasi-surely}, that is, $P$-almost surely for all $P\in\mathcal{P}$. For further details, the reader may refer to the book of \citet*{peng2019}.
\par We model the forward rate as a diffusion process in the spirit of the HJM methodology. We denote by $f_t(T)$ the forward rate with maturity $T$ at time $t$ for $t\leq T\leq\bar{T}$, where $\bar{T}<\infty$ is a fixed terminal time. We assume that the forward rate process $f(T)=(f_t(T))_{0\leq t\leq T}$, for all $T\leq\bar{T}$, evolves according to the dynamics
\begin{align*}
f_t(T)=f_0(T)+\int_0^t\alpha_u(T)du+\sum_{i=1}^d\int_0^t\beta_u^i(T)dB_u^i+\sum_{i=1}^d\int_0^t\gamma_u^i(T)d\langle B^i\rangle_u
\end{align*}
for some initial integrable forward curve $f_0:[0,\bar{T}]\rightarrow\mathbb{R}$ and sufficiently regular processes $\alpha(T)=(\alpha_t(T))_{0\leq t\leq\bar{T}}$, $\beta^i(T)=(\beta_t^i(T))_{0\leq t\leq\bar{T}}$, and $\gamma^i(T)=(\gamma_t^i(T))_{0\leq t\leq\bar{T}}$ to be specified. The difference compared to the classical HJM model without volatility uncertainty is that there are additional drift terms depending on the quadratic variation processes of the $G$-Brownian motion. We need the additional drift terms in order to obtain an arbitrage-free model as it is described below. However, due to the uncertainty about the volatility, the quadratic variation processes are uncertain, which cannot be included in the first drift term. Thus, we add additional drift terms to the dynamics of the forward rate. More details on this can be found in the companion paper \citep*[Section 2]{holzermann2021'}.
\par The forward rate determines the remaining quantities on the bond market. The bond market consists of zero-coupon bonds for all maturities in the time horizon and the money-market account. The zero-coupon bonds, denoted by $P(T)=(P_t(T))_{0\leq t\leq T}$ for $T\leq\bar{T}$, are defined by
\begin{align*}
P_t(T):=\exp\biggl(-\int_t^Tf_t(s)ds\biggr),
\end{align*}
and the money-market account, denoted by $M=(M_t)_{0\leq t\leq\bar{T}}$, is given by
\begin{align*}
M_t:=\exp\biggl(\int_0^tr_sds\biggr),
\end{align*}
where $r=(r_t)_{0\leq t\leq\bar{T}}$ denotes the short rate process, defined by $r_t:=f_t(t)$. We use the money-market account as a num\'eraire---that is, we focus on the discounted bonds, which are denoted by $\tilde{P}(T)=(\tilde{P}_t(T))_{0\leq t\leq T}$ for $T\leq\bar{T}$ and given by
\begin{align*}
\tilde{P}_t(T):=M_t^{-1}P_t(T).
\end{align*}
\par We model the forward rate in such a way that the related bond market is arbitrage-free. That means, we assume that the forward rate satisfies a suitable drift condition, which implies the absence of arbitrage. In particular, we directly model the forward rate in a risk-neutral way in order to avoid technical difficulties due to a migration to a risk-neutral framework. More specifically, for all $T$, we assume that the drift terms $\alpha(T)$ and $\gamma^i(T)$, for all $i$, are defined by
\begin{align*}
\alpha_t(T):=0,\quad\gamma_t^i(T):=\beta_t^i(T)b_t^i(T),
\end{align*}
respectively, where the process $b^i(T)=(b_t^i(T))_{0\leq t\leq\bar{T}}$ is defined by
\begin{align*}
b_t^i(T):=\int_t^T\beta_t^i(s)ds.
\end{align*}
Under suitable regularity assumptions on $T\mapsto\beta^i(T)$, we can then show that the discounted bonds are symmetric $G$-martingales under $\hat{\mathbb{E}}$, which implies that the bond market is arbitrage-free \citep*[Theorem 3.1]{holzermann2021'}. As mentioned above, this shows that we need the additional drift terms in the forward rate dynamics to obtain an arbitrage-free model.
\par In order to achieve a sufficient degree of regularity and to derive pricing formulas for derivative contracts, we use a deterministic diffusion coefficient. We assume that $\beta^i$, for all $i$, is a continuous function mapping from $[0,\bar{T}]\times[0,\bar{T}]$ into $\mathbb{R}$. Then for each $T$, the processes $\beta^i(T)$ and $b^i(T)$, for all $i$, are bounded processes in $M_G^p(0,\bar{T})$ for all $p<\infty$. The space $M_G^p(0,\bar{T})$ is the space of admissible integrands for stochastic integrals related to a $G$-Brownian motion. The assumption ensures that the forward rate is sufficiently regular to apply the result from above. In addition, it enables us to obtain specific pricing formulas for common interest rate derivatives. This is similar to the classical case without volatility uncertainty, in which it is possible to obtain analytical pricing formulas by assuming that the diffusion coefficient is deterministic. So the present model corresponds to an HJM model with a normally distributed forward rate.

\section{Risk-Neutral Valuation}\label{risk-neutral valuation}
Now we extend the bond market to an additional contract, for which we want to find a price. A typical contract in fixed income markets consists of a stream of cashflows; so we consider a contract, denoted by $X$, that has a payoff of $\xi_i$ at each time $T_i$ for all $i=0,1,...,N$, where $0<T_0<T_1<...<T_N=\bar{T}$ is the tenor structure. The price at time $t$ of such a contract is denoted by $X_t$ for all $t\leq\bar{T}$. As for the bonds, we consider the discounted payoff $\tilde{X}$, defined by
\begin{align*}
\tilde{X}:=\sum_{i=0}^NM_{T_i}^{-1}\xi_i,
\end{align*}
and the discounted price $\tilde{X}_t$ for $t\leq\bar{T}$, which is defined by
\begin{align*}
\tilde{X}_t:=M_t^{-1}X_t.
\end{align*}
We assume that $M_{T_i}^{-1}\xi_i\in L_G^2(\Omega_{T_i})$ for all $i=0,1,...,N$ for $X$ to be regular enough.
\par The pricing of contracts in the presence of volatility uncertainty differs from the traditional approach. Classical arbitrage pricing theory suggests that prices of contracts are determined by computing the expected discounted payoff under the risk-neutral measure. In the presence of volatility uncertainty, we call $\hat{\mathbb{E}}$ the \textit{risk-neutral sublinear expectation}, corresponding to the expectation under the risk-neutral measure in the classical case, since the discounted bonds are symmetric $G$-martingales under $\hat{\mathbb{E}}$. Compared to the classical case, the important difference in the case of volatility uncertainty is that the risk-neutral sublinear expectation is nonlinear. In particular, it holds
\begin{align}\label{inequality}
\hat{\mathbb{E}}[\tilde{X}]\geq-\hat{\mathbb{E}}[-\tilde{X}],
\end{align}
that is, the upper expectation does not necessarily coincide with the lower expectation. Thus, we distinguish between symmetric and asymmetric contracts; we consider two contracts: a contract $X^S$, which has a symmetric payoff, and a contract $X^A$, which has an asymmetric payoff. Strictly speaking, this means that $\tilde{X}^S$ satisfies \eqref{inequality} with equality and for $\tilde{X}^A$, the inequality \eqref{inequality} is strict. Of course, the discounted payoffs $\tilde{X}^S$ and $\tilde{X}^A$ are defined as above by considering different payoffs $\xi_i^S$ and $\xi_i^A$ for all $i$, respectively. The related prices are denoted by $X_t^S$ and $\tilde{X}_t^S$ and $X_t^A$ and $\tilde{X}_t^A$ for all $t$, respectively.
\par We determine the prices of contracts by using the risk-neutral sublinear expectation to either obtain the price of a contract or the upper and the lower bound for the price. In the case of a symmetric payoff, we proceed as in the classical case without volatility uncertainty and choose the expected discounted payoff as the price for the contract. In the case of an asymmetric payoff, we use the upper and the lower expectation as bounds for the price, which is a typical approach in the literature on model uncertainty and yields a range of possible prices. Hence, we assume that
\begin{align*}
\tilde{X}_t^S=\hat{\mathbb{E}}_t[\tilde{X}^S]
\end{align*}
for all $t$, where $\hat{\mathbb{E}}_t$ denotes the conditional $G$-expectation, and
\begin{align*}
\hat{\mathbb{E}}[\tilde{X}^A]>\tilde{X}_0^A>-\hat{\mathbb{E}}[-\tilde{X}^A].
\end{align*}
Since $X^S$ has a symmetric payoff, by the martingale representation theorem for symmetric $G$-martingales \citep*[Theorem 4.8]{song2011}, there exists a process $H=(H_t^1,...,H_t^d)_{0\leq t\leq\bar{T}}$ in $M_G^2(0,\bar{T};\mathbb{R}^d)$ such that for all $t$,
\begin{align*}
\tilde{X}_t^S=\tilde{X}_0^S+\sum_{i=1}^d\int_0^tH_u^idB_u^i.
\end{align*}
The latter ensures that the portfolio value (defined below) is well-posed. The reason why we only impose assumptions on the price of the asymmetric contract at time $0$ is described below.
\par In order to show that this pricing procedure yields no-arbitrage prices, we introduce the notion of trading strategies related to the extended bond market and a suitable notion of arbitrage. We allow the agents in the market to trade a finite number of bonds. The symmetric contract can be traded dynamically, but we only allow static trading strategies for the asymmetric contract. Therefore, we do not impose assumptions on $\tilde{X}_t^A$ for $t>0$. The assumption that the asymmetric contract can only be traded statically might seem restrictive. This is a common assumption in the literature on robust finance, since it is important for excluding arbitrage. In this case, the assumption is also reasonable, since most contracts in fixed income markets are traded over-the-counter.
\begin{df}\label{admissible market strategy}
An admissible market strategy is a quadruple $(\pi,\pi^S,\pi^A,\tau)$ consisting of a bounded process $\pi=(\pi_t^1,...,\pi_t^n)_{0\leq t\leq\bar{T}}$ in $M_G^2(0,\bar{T};\mathbb{R}^n)$, a bounded process $\pi^S=(\pi_t^S)_{0\leq t\leq\bar{T}}$ in $M_G^2(0,\bar{T})$, a constant $\pi^A\in\mathbb{R}$, and a vector $\tau=(\tau_1,...,\tau_n)\in[0,\bar{T}]^n$ for some $n\in\mathbb{N}$. The corresponding portfolio value at terminal time is given by
\begin{align}\label{portfolio value}
\tilde{v}(\pi,\pi^S,\pi^A,\tau):=\sum_{i=1}^n\int_0^{\tau_i}\pi^i_td\tilde{P}_t(\tau_i)+\int_0^{\bar{T}}\pi_t^Sd\tilde{X}_t^S+\pi^A(\tilde{X}^A-\tilde{X}_0^A).
\end{align}
\end{df}
\noindent The three terms on the right-hand side of \eqref{portfolio value} correspond to the gains from trading a finite number of bonds, the symmetric contract, and the asymmetric contract, respectively. The assumptions on the processes ensure that the integrals in \eqref{portfolio value} are well-defined. In addition, we use the quasi-sure definition of arbitrage, which is commonly used in the literature on model uncertainty \citep*{biaginibouchardkardarasnutz2017,bouchardnutz2015}.
\begin{df}\label{arbitrage strategy}
An admissible market strategy $(\pi,\pi^S,\pi^A,\tau)$ is an arbitrage strategy if
\begin{align*}
\tilde{v}(\pi,\pi^S,\pi^A,\tau)\geq{}&0\quad\text{quasi-surely},
\\P\big(\tilde{v}(\pi,\pi^S,\pi^A,\tau)>0\big)>{}&0\quad\text{for at least one }P\in\mathcal{P}.
\end{align*}
We say that the extended bond market is arbitrage-free if there is no arbitrage strategy.
\end{df}
\par The following proposition shows that we can use the risk-neutral sublinear expectation as a pricing measure as described above, since the extended bond market is arbitrage-free under the assumptions on the prices of the symmetric and the asymmetric contract.
\begin{prp}\label{no arbitrage}
The extended bond market is arbitrage-free.
\end{prp}
\begin{proof}
We assume that there exists an arbitrage strategy $(\pi,\pi^S,\pi^A,\tau)$ and show that this yields a contradiction. We only examine the case in which $X^A$ is traded, i.e., it holds $\pi^A\neq0$; if $\pi^A=0$, the proof is similar to showing that the bond market is arbitrage-free \citep*[Proposition 4.1]{holzermann2021}. By the definition of arbitrage, it holds $\tilde{v}(\pi,\pi^S,\pi^A,\tau)\geq0$. Then the monotonicity of $\hat{\mathbb{E}}$ implies that
\begin{align*}
\hat{\mathbb{E}}\biggl[\sum_{i=1}^n\int_0^{\tau_i}\pi^i_td\tilde{P}_t(\tau_i)+\int_0^{\bar{T}}\pi_t^Sd\tilde{X}_t^S\biggr]\geq\hat{\mathbb{E}}[-\pi^A(\tilde{X}^A-\tilde{X}_0^A)].
\end{align*}
Due to the sublinearity of $\hat{\mathbb{E}}$ and the fact that the discounted bonds and the discounted price process of the symmetric contract are symmetric $G$-martingales under $\hat{\mathbb{E}}$, we have
\begin{align*}
\hat{\mathbb{E}}\biggl[\sum_{i=1}^n\int_0^{\tau_i}\pi^i_td\tilde{P}_t(\tau_i)+\int_0^{\bar{T}}\pi_t^Sd\tilde{X}_t^S\biggr]\leq0.
\end{align*}
Furthermore, if we use the properties of $\hat{\mathbb{E}}$ and the assumption on $\tilde{X}_0^A$, we get
\begin{align*}
\hat{\mathbb{E}}[-\pi^A(\tilde{X}^A-\tilde{X}_0^A)]=(\pi^A)^+(\hat{\mathbb{E}}[-\tilde{X}^A]+\tilde{X}_0^A)+(\pi^A)^-(\hat{\mathbb{E}}[\tilde{X}^A]-\tilde{X}_0^A)>0.
\end{align*}
Combining the previous steps, we obtain a contradiction.
\end{proof}
\begin{rem}\label{remark on pricing}
As a consequence of Proposition \ref{no arbitrage}, we can reduce the problem of pricing a contract to evaluating the upper and the lower expectation of its discounted payoff. Then the upper and the lower expectation yield the price of the contract if both coincide or otherwise, the upper and the lower bound for the price, respectively.
\end{rem}
\begin{rem}\label{remark on hedging}
The pricing-hedging duality in the presence of volatility uncertainty shows that prices differing from the pricing procedure in this section lead to arbitrage. Strictly speaking, Proposition \ref{no arbitrage} only shows that the pricing procedure in this section yields no-arbitrage prices but not that other prices create arbitrage opportunities. From the pricing-hedging duality under model uncertainty (see \citet*[Theorem 3.6]{vorbrink2014} for volatility uncertainty), we can deduce that the upper, respectively lower, expectation corresponds to the smallest superhedging, respectively highest subhedging, price. Hence, there exists an arbitrage strategy if the price of a contract is greater, respectively less, than the upper, respectively lower, expectation of its discounted payoff.
\par Moreover, the pricing-hedging duality under volatility uncertainty provides a (super)hedging strategy for contracts in fixed income markets. From the pricing-hedging duality, we can additionally infer that symmetric contracts can be hedged while asymmetric contracts have to be superhedged. In fact, the prices or the pricing bounds of many interest rate derivatives are given by the prices from the corresponding HJM model without volatility uncertainty (see Section \ref{common interest rate derivatives}); thus, we can use traditional hedging strategies to (super)hedge typical fixed income derivatives in the presence of volatility uncertainty.
\end{rem}

\section{Pricing Single Cashflows}\label{pricing single cashflows}
In the classical case without volatility uncertainty, discounted cashflows are priced under the forward measure. Evaluating the expectation of a discounted cashflow related to an interest rate derivative can be very elaborate; this is due to the fact that the discount factor---in addition to the cashflows---is stochastic. The common way to avoid this issue is the forward measure approach. The forward measure, which was introduced by \citet*{geman1989}, is equivalent to the pricing measure and defined by choosing a particular density process. The density process is defined in such a way that the expectation of a discounted cashflow under the risk-neutral measure can be rewritten as the expectation of the cashflow under the forward measure, discounted by a zero-coupon bond. Thus, by changing the measure, we can replace the stochastic discount factor by the current bond price (which is already determined by the model).
\par In the presence of volatility uncertainty, we define a counterpart of the forward measure, termed \textit{forward sublinear expectation}, to simplify the pricing of discounted cashflows. In contrast to the forward measure approach, we define the forward sublinear expectation by a $G$-BSDE.
\begin{df}\label{forward sublinear expectation}
For $\xi\in L_G^p(\Omega_T)$ with $p>1$ and $T\leq\bar{T}$, we define the $T$-forward sublinear expectation $\hat{\mathbb{E}}^T$ by $\hat{\mathbb{E}}_t^T[\xi]:=Y_t^{T,\xi}$, where $Y^{T,\xi}=(Y_t^{T,\xi})_{0\leq t\leq T}$ solves the $G$-BSDE
\begin{align*}
Y_t^{T,\xi}=\xi-\sum_{i=1}^d\int_t^Tb_u^i(T)Z_u^id\langle B^i\rangle_u-\sum_{i=1}^d\int_t^TZ_u^idB_u^i-(K_T-K_t).
\end{align*}
\end{df}
\noindent By Theorem 5.1 of \citet*{hujipengsong2014}, the forward sublinear expectation is a time consistent sublinear expectation. We refer to the paper of \citet*{hujipengsong2014} for further details related to $G$-BSDEs.
\par The forward sublinear expectation corresponds to the expectation under the forward measure. This can be deduced from the explicit solution to the $G$-BSDE defining the forward sublinear expectation. For $T\leq\bar{T}$, we define the process $X^T=(X_t^T)_{0\leq t\leq T}$ by
\begin{align*}
X_t^T:=\tfrac{\tilde{P}_t(T)}{P_0(T)}.
\end{align*}
The process $X^T$ is the density used to define the forward measure. One can verify that $X^T$ satisfies the dynamics
\begin{align*}
X_t^T=1-\sum_{i=1}^d\int_0^tb_u^i(T)X_u^TdB_u^i
\end{align*}
\citep*[Proposition 3.1]{holzermann2021'}. By Theorem 3.2 of \citet*{hujipengsong2014}, we know that the process $Y^{T,\xi}$ is given by
\begin{align*}
Y_t^{T,\xi}=(X_t^T)^{-1}\hat{\mathbb{E}}_t[X_T^T\xi].
\end{align*}
Thus, we basically arrive at the same expression as in the classical definition of the forward measure.
\par We obtain the following preliminary results related to the forward sublinear expectation, which simplify the pricing of discounted cashflows. Similar to the classical case, we find that pricing a discounted cashflow reduces to determining the forward sublinear expectation of the cashflow, which is then discounted with the bond price. Furthermore, there is a relation between forward sublinear expectations with different maturities, and the forward rate process and the forward price process, which is denoted by $X^{T,\tilde{T}}=(X_t^{T,\tilde{T}})_{0\leq t\leq T\wedge\tilde{T}}$ for $T,\tilde{T}\leq\bar{T}$ and defined by
\begin{align*}
X_t^{T,\tilde{T}}:=\tfrac{P_t(\tilde{T})}{P_t(T)},
\end{align*}
are symmetric $G$-martingales under the $T$-forward sublinear expectation.
\begin{prp}\label{preliminary results}
Let $\xi\in L_G^p(\Omega_T)$ with $p>1$ and let $t\leq T,\tilde{T}\leq\bar{T}$. Then we have the following properties.
\begin{enumerate}
\item[(i)] It holds
\begin{align*}
M_t\hat{\mathbb{E}}_t[M_T^{-1}\xi]=P_t(T)\hat{\mathbb{E}}_t^T[\xi].
\end{align*}
\item[(ii)] For $T\leq\tilde{T}$, it holds
\begin{align*}
P_t(\tilde{T})\hat{\mathbb{E}}_t^{\tilde{T}}[\xi]=P_t(T)\hat{\mathbb{E}}_t^T[P_T(\tilde{T})\xi].
\end{align*}
\item[(iii)] The process $f(T)$ is a symmetric $G$-martingale under $\hat{\mathbb{E}}^T$.
\item[(iv)] The process $X^{T,\tilde{T}}$ satisfies $X_t^{T,\tilde{T}}\in L_G^p(\Omega_t)$ for all $p<\infty$ and
\begin{align*}
X_t^{T,\tilde{T}}=X_0^{T,\tilde{T}}-\sum_{i=1}^d\int_0^t\sigma_u^i(T,\tilde{T})X_u^{T,\tilde{T}}dB_u^i-\sum_{i=1}^d\int_0^t\sigma_u^i(T,\tilde{T})X_u^{T,\tilde{T}}b_u^i(T)d\langle B^i\rangle_u,
\end{align*}
where $\sigma^i(T,\tilde{T})=(\sigma_t^i(T,\tilde{T}))_{0\leq t\leq T\wedge\tilde{T}}$, for all $i$, is defined by
\begin{align*}
\sigma_t^i(T,\tilde{T}):=b_t^i(\tilde{T})-b_t^i(T),
\end{align*}
and $X^{T,\tilde{T}}$ is a symmetric $G$-martingale under $\hat{\mathbb{E}}^T$.
\end{enumerate}
\end{prp}
\begin{proof}
Part $(i)$ follows by a simple calculation; we have
\begin{align*}
M_t\hat{\mathbb{E}}_t[M_T^{-1}\xi]=P_t(T)M_t\tfrac{P_0(T)}{P_t(T)}\hat{\mathbb{E}}_t[M_T^{-1}\tfrac{P_T(T)}{P_0(T)}\xi]
=P_t(T)(X_t^T)^{-1}\hat{\mathbb{E}}_t[X_T^T\xi]
=P_t(T)\hat{\mathbb{E}}_t^T[\xi].
\end{align*}
\par To show part $(ii)$, we use some properties of $G$-BSDEs. By Definition \ref{forward sublinear expectation}, we have $\hat{\mathbb{E}}_t^{\tilde{T}}[\xi]=Y_t^{\tilde{T},\xi}$, where $Y^{\tilde{T},\xi}$ solves
\begin{align*}
Y_t^{\tilde{T},\xi}=\xi-\sum_{i=1}^d\int_t^{\tilde{T}}b_u^i(\tilde{T})Z_u^id\langle B^i\rangle_u-\sum_{i=1}^d\int_t^{\tilde{T}}Z_u^idB_u^i-(K_{\tilde{T}}-K_t).
\end{align*}
Since $\xi\in L_G^p(\Omega_T)$, the process $Y^{\tilde{T},\xi}$ also solves the $G$-BSDE
\begin{align*}
Y_t^{\tilde{T},\xi}=\xi-\sum_{i=1}^d\int_t^Tb_u^i(\tilde{T})Z_u^id\langle B^i\rangle_u-\sum_{i=1}^d\int_t^TZ_u^idB_u^i-(K_T-K_t).
\end{align*}
By Theorem 3.2 of \citet*{hujipengsong2014}, the solution to the latter is given by
\begin{align*}
Y_t^{\tilde{T},\xi}=(X_t^{\tilde{T}})^{-1}\hat{\mathbb{E}}_t[X_T^{\tilde{T}}\xi].
\end{align*}
Moreover, for each $t\leq T$, we have $X_t^{\tilde{T}}=X_t^{T,\tilde{T}}X_0^{\tilde{T},T}X_t^T$. Hence, we obtain
\begin{align*}
\hat{\mathbb{E}}_t^{\tilde{T}}[\xi]=X_t^{\tilde{T},T}X_0^{T,\tilde{T}}(X_t^T)^{-1}\hat{\mathbb{E}}_t[X_T^{T,\tilde{T}}X_0^{\tilde{T},T}X_T^T\xi]=X_t^{\tilde{T},T}\hat{\mathbb{E}}_t^T[X_T^{T,\tilde{T}}\xi],
\end{align*}
which proves part $(ii)$.
\par For part $(iii)$, we use the Girsanov transformation for $G$-Brownian motion from \citet*{hujipengsong2014}. We define the process $B^T=(B_t^{1,T},...,B_t^{d,T})_{0\leq t\leq T}$ by
\begin{align*}
B_t^{i,T}:=B_t^i+\int_0^tb_u^i(T)d\langle B^i\rangle_u.
\end{align*}
Then $B^T$ is a $G$-Brownian motion under $\hat{\mathbb{E}}^T$ \citep*[Theorems 5.2, 5.4]{hujipengsong2014}. Since the dynamics of the forward rate are given by
\begin{align*}
f_t(T)=f_0(T)+\sum_{i=1}^d\int_0^t\beta_u^i(T)dB_u^i+\sum_{i=1}^d\int_0^t\beta_u^i(T)b_u^i(T)d\langle B^i\rangle_u,
\end{align*}
the forward rate is a symmetric $G$-martingale under $\hat{\mathbb{E}}^T$.
\par To obtain part $(iv)$, we first show that $X_t^{T,\tilde{T}}\in L_G^p(\Omega_t)$ for all $p<\infty$ by using the representation of the space $L_G^p(\Omega_t)$ from \citet*{denishupeng2011} and a proof similar to the proof of Proposition 5.10 from \citet*{osuka2013}. The space $L_G^p(\Omega_t)$ consists of all Borel measurable random variables $X$ which have a quasi-continuous version and satisfy $\lim_{n\rightarrow\infty}\hat{\mathbb{E}}[\vert X\vert^p1_{\{\vert X\vert>n\}}]=0$ \citep*[Proposition 6.3.2]{peng2019}. One can show that
\begin{align*}
X_t^{T,\tilde{T}}=X_0^{T,\tilde{T}}\exp\biggl(-\sum_{i=1}^d\int_0^t\sigma_u^i(T,\tilde{T})dB_u^i-\sum_{i=1}^d\int_0^t\big(\tfrac{1}{2}\sigma_u^i(T,\tilde{T})^2+\sigma_u^i(T,\tilde{T})b_u^i(T)\big)d\langle B^i\rangle_u\biggr)
\end{align*}
\citep*[Lemma 3.1]{holzermann2021'}. Since $\sigma^i(T,\tilde{T})$ and $b^i(T)$, for all $i$, are bounded processes in $M_G^p(0,\bar{T})$ for all $p<\infty$, we already know that $X_t^{T,\tilde{T}}$ is measurable and has a quasi-continuous version. Now we show that $\hat{\mathbb{E}}[\vert X_t^{T,\tilde{T}}\vert^{\tilde{p}}]<\infty$ for $\tilde{p}>p$, which implies $\lim_{n\rightarrow\infty}\hat{\mathbb{E}}[\vert X\vert^p1_{\{\vert X\vert>n\}}]=0$. By H\"older's inequality, for $\tilde{p}>p$ and $\tilde{q}>1$, we have
\begin{align*}
\hat{\mathbb{E}}[\vert X_t^{T,\tilde{T}}\vert^{\tilde{p}}]\leq{}&X_0^{T,\tilde{T}}\hat{\mathbb{E}}\biggl[\exp\biggl(-\tilde{p}\tilde{q}\sum_{i=1}^d\int_0^t\sigma_u^i(T,\tilde{T})dB_u^i-\tfrac{1}{2}(\tilde{p}\tilde{q})^2\sum_{i=1}^d\int_0^t\sigma_u^i(T,\tilde{T})^2d\langle B^i\rangle_u\biggr)\biggr]^\frac{1}{\tilde{q}}
\\&\times\hat{\mathbb{E}}\biggl[\exp\biggl(\tfrac{\tilde{p}\tilde{q}}{\tilde{q}-1}\sum_{i=1}^d\int_0^t\big(\tfrac{1}{2}(\tilde{p}\tilde{q}-1)\sigma_u^i(T,\tilde{T})^2-\sigma_u^i(T,\tilde{T})b_u^i(T)\big)d\langle B^i\rangle_u\biggr)\biggr]^\frac{\tilde{q}-1}{\tilde{q}}.
\end{align*}
The two terms on the right-hand side are finite. The second term is finite since $\sigma^i(T,\tilde{T})$ and $b^i(T)$ are bounded for all $i$. By the same argument, we have
\begin{align*}
\hat{\mathbb{E}}\biggl[\exp\biggl(\tfrac{1}{2}(\tilde{p}\tilde{q})^2\sum_{i=1}^d\int_0^t\sigma_u^i(T,\tilde{T})^2d\langle B^i\rangle_u\biggr)\biggr]<\infty.
\end{align*}
Then we can use Novikov's condition to show that the first term is finite, since the exponential inside the sublinear expectation is a martingale under each $P\in\mathcal{P}$.
\par Using It\^o's formula for $G$-Brownian motion from \citet*{lipeng2011} and the Girsanov transformation of \citet*{hujipengsong2014} completes the proof. We have
\begin{align*}
X_t^{T,\tilde{T}}=X_0^{T,\tilde{T}}-\sum_{i=1}^d\int_0^t\sigma_u^i(T,\tilde{T})X_u^{T,\tilde{T}}dB_u^i-\sum_{i=1}^d\int_0^t\sigma_u^i(T,\tilde{T})X_u^{T,\tilde{T}}b_u^i(T)d\langle B^i\rangle_u
\end{align*}
by It\^o's formula \citep*[Theorem 5.4]{lipeng2011}. Moreover, since $\sigma^i(T,\tilde{T})$ and $b^i(T)$, for all $i$, are bounded processes in $M_G^p(0,\bar{T})$ for all $p<\infty$, one can then show that $X^{T,\tilde{T}}$ belongs to $M_G^p(0,\bar{T})$ for all $p<\infty$ \citep*[Proposition B.1]{holzermann2021'}. As in the proof of part $(iii)$, the Girsanov transformation then implies that $X^{T,\tilde{T}}$ is a symmetric $G$-martingale under $\hat{\mathbb{E}}^T$.
\end{proof}
\par Due to Proposition \ref{preliminary results} $(iii)$, we obtain a robust version of the expectations hypothesis as a by-product. The traditional expectations hypothesis states that forward rates reflect the expectation of future short rates. In the classical case without volatility uncertainty, the forward rate is a martingale under the forward measure; therefore, the expectations hypothesis holds true under the forward measure. In our case, we obtain a much stronger version---called \textit{robust expectations hypothesis}. The forward rate is a symmetric $G$-martingale under the forward sublinear expectation; thus, the forward rate reflects the upper and the lower expectation of the short rate.
\begin{cor}\label{robust expectations hypothesis}
The forward rate satisfies the robust expectations hypothesis under the forward sublinear expectation---that is, for $t\leq T\leq\bar{T}$, it holds
\begin{align*}
\hat{\mathbb{E}}_t^T[r_T]=f_t(T)=-\hat{\mathbb{E}}_t^T[-r_T].
\end{align*}
\end{cor}
\noindent So in particular, the forward rate reflects the expectation of the short rate in each possible scenario for the volatility.
\par Next, we consider an option written on forward prices. The cashflows of most nonlinear contracts in fixed income markets can be written as bond options or, equivalently, as options on forward prices (see, e.g., Subsection \ref{swaptions}). Thus, we now consider the case when the payoff is given by a function depending on a selection of forward prices for different maturities: for $n\in\mathbb{N}$, let $\xi$ be defined by
\begin{align}\label{definition of xi}
\xi:=\varphi\big((X_{t_1}^{T,t_i})_{i=1}^n\big)
\end{align}
for a function $\varphi:\mathbb{R}^n\rightarrow\mathbb{R}$ and a tenor structure $0<t_1<...<t_n\leq\bar{T}$ with $t_1\leq T\leq\bar{T}$.
\par The price of such an option is characterized by a nonlinear PDE. By using a nonlinear version of the Feynman-Kac formula, we find that evaluating the forward sublinear expectation of the payoff reduces to solving a nonlinear PDE.
\begin{prp}\label{pricing a bond option}
Let $\xi$ be given by \eqref{definition of xi}. If $\varphi$ satisfies
\begin{align}\label{estimate}
\vert\varphi(x)-\varphi(y)\vert\leq C(1+\vert x\vert^m+\vert y\vert^m)\vert x-y\vert
\end{align}
for a positive integer $m$ and a constant $C>0$, then for $t\leq t_1$,
\begin{align*}
\hat{\mathbb{E}}_t^T[\xi]=u\big(t,(X_t^{T,t_i})_{i=1}^n\big),
\end{align*}
where $u:[0,t_1]\times\mathbb{R}^n\rightarrow\mathbb{R}$ is the unique viscosity solution to the nonlinear PDE
\begin{align}\label{pde}
\begin{split}
\partial_tu+\tfrac{1}{2}\sum_{j=1}^d\Big(\overline{\sigma}_j^2\big(\sigma^j(t,x)D_{xx}^2u\,\sigma^j(t,x)'\big)^+-\underline{\sigma}_j^2\big(\sigma^j(t,x)D_{xx}^2u\,\sigma^j(t,x)'\big)^-\Big)={}&0,
\\u(t_1,x)={}&\varphi(x)
\end{split}
\end{align}
and $\sigma^j(t,x):=(\sigma_t^j(T,t_i)x_i)_{i=1}^n$.
\end{prp}
\begin{proof}
We show the assertion by using the nonlinear Feynman-Kac formula of \citet*{hujipengsong2014}. With Proposition \ref{preliminary results} $(iv)$ and inequality \eqref{estimate}, one can show that $\xi$ belongs to $L_G^p(\Omega_{t_1})\subset L_G^p(\Omega_T)$ with $p>1$. By Definition \ref{forward sublinear expectation}, we have $\hat{\mathbb{E}}_t^T[\xi]=Y_t^{T,\xi}$, where $Y^{T,\xi}=(Y_t^{T,\xi})_{0\leq t\leq T}$ solves the $G$-BSDE
\begin{align*}
Y_t^{T,\xi}=\xi-\sum_{i=1}^d\int_t^Tb_u^i(T)Z_u^id\langle B^i\rangle_u-\sum_{i=1}^d\int_t^TZ_u^idB_u^i-(K_T-K_t).
\end{align*}
Since $\xi\in L_G^p(\Omega_{t_1})$, the process $Y^{T,\xi}$ also solves the $G$-BSDE
\begin{align*}
Y_t^{T,\xi}=\xi-\sum_{i=1}^d\int_t^{t_1}b_u^i(T)Z_u^id\langle B^i\rangle_u-\sum_{i=1}^d\int_t^{t_1}Z_u^idB_u^i-(K_{t_1}-K_t),
\end{align*}
where $\varphi$ satisfies \eqref{estimate}. From Proposition \ref{preliminary results} $(iv)$, we deduce the dynamics and the regularity of $X^{T,t_i}$ for all $i=1,...,n$. Then, by Theorems 4.4 and 4.5 of \citet*{hujipengsong2014}, we have $Y_t^{T,\xi}=u(t,(X_t^{T,t_i})_{i=1}^n)$, where $u:[0,t_1]\times\mathbb{R}^n\rightarrow\mathbb{R}$ is the unique viscosity solution to \eqref{pde}.
\end{proof}
\par When the option's payoff is additionally convex or concave, the price is characterized by the price from the corresponding HJM model without volatility uncertainty. If the payoff function is convex, respectively concave, then we can show that the forward sublinear expectation corresponds to the linear expectation of the payoff when the dynamics of the forward prices are driven by a standard Brownian motion with constant volatility $\overline{\sigma}$, respectively $\underline{\sigma}$.
\begin{prp}\label{pricing a convex bond option}
Let $\xi$ be given by \eqref{definition of xi}. If $\varphi$ is convex and satisfies \eqref{estimate}, then
\begin{align*}
\hat{\mathbb{E}}_t^T[\xi]=u^{\overline{\sigma}}\big(t,(X_t^{T,t_i})_{i=1}^n\big),
\end{align*}
for $t\leq t_1$, where the function $u^\sigma:[0,t_1]\times\mathbb{R}^n\rightarrow\mathbb{R}$, for $\sigma\in\Sigma$, is defined by
\begin{align*}
u^\sigma(t,x):=\mathbb{E}_{P_0}\big[\varphi\big((X_{t_1}^i)_{i=1}^n\big)\big]
\end{align*}
and the process $X^i=(X_s^i)_{t\leq s\leq t_1}$, for all $i=1,...,n$, is given by
\begin{align*}
X_s^i=x_i-\sum_{j=1}^d\int_t^s\sigma_u^j(T,t_i)X_u^i\sigma_jdB_u^j.
\end{align*}
If $\varphi$ is concave instead of convex, then for $t\leq t_1$,
\begin{align*}
\hat{\mathbb{E}}_t^T[\xi]=u^{\underline{\sigma}}\big(t,(X_t^{T,t_i})_{i=1}^n\big).
\end{align*}
\end{prp}
\begin{proof}
We show that $u^{\overline{\sigma}}$ solves the nonlinear PDE \eqref{pde} and apply Proposition \ref{pricing a bond option} to prove the first assertion; the proof of the second assertion is analogous. By the classical Feynman-Kac formula, we know that $u_{\overline{\sigma}}$ solves
\begin{align*}
\partial_tu+\tfrac{1}{2}\sum_{j=1}^d{\overline{\sigma}}_j^2\sigma^j(t,x)D_{xx}^2u\,\sigma^j(t,x)'=0,\quad u(t_1,x)=\varphi(x).
\end{align*}
In addition, the convexity of $\varphi$ implies that $u^{\overline{\sigma}}(t,\cdot)$ is convex for each $t$; thus,
\begin{align*}
\sigma^j(t,x)D_{xx}^2u^{\overline{\sigma}}\,\sigma^j(t,x)'\geq0
\end{align*}
for all $j=1,...,d$. Therefore, one can verify that $u^{\overline{\sigma}}$ solves \eqref{pde}. Then the claim follows by Proposition \ref{pricing a bond option}.
\end{proof}

\section{Pricing a Stream of Cashflows}\label{pricing a stream of cashflows}
Due to the nonlinearity of the pricing measure, in general, we cannot price interest rate derivatives by pricing each cashflow separately. As in Section \ref{risk-neutral valuation}, we consider a contract consisting of a stream of cashflows, which we denote by $X$. Then the discounted payoff is given by
\begin{align*}
\tilde{X}=\sum_{i=0}^NM_{T_i}^{-1}\xi_i
\end{align*}
for a tenor structure $0<T_0<T_1<...<T_N=\bar{T}$ and $\xi_i\in L_G^p(\Omega_{T_i})$ with $p>1$ for all $i$. In order to price the contract, we are interested in $\hat{\mathbb{E}}[\tilde{X}]$ and $-\hat{\mathbb{E}}[-\tilde{X}]$. When there is no volatility uncertainty, we can simply price the contract by pricing each cashflow individually, since the pricing measure is linear in that case. However, in the presence of volatility uncertainty, the pricing measure is sublinear, which implies
\begin{align*}
\hat{\mathbb{E}}[\tilde{X}]\leq\sum_{i=0}^N\hat{\mathbb{E}}[M_{T_i}^{-1}\xi_i],\quad-\hat{\mathbb{E}}[-\tilde{X}]\geq\sum_{i=0}^N-\hat{\mathbb{E}}[-M_{T_i}^{-1}\xi_i].
\end{align*}
Therefore, if we price each cashflow separately, we possibly only obtain an upper, respectively lower, bound for the upper, respectively lower, bound of the price---which does not yield much information about the price of the contract.
\par If the contract has symmetric cashflows, then it has a single price and we can determine the price by pricing each of its cashflows individually. For contracts with symmetric cashflows, the upper expectation coincides with the lower expectation of the discounted payoff, and we obtain both by computing the forward sublinear expectation of each cashflow separately.
\begin{lem}\label{pricing contracts with symmetric cashflows}
If $\xi_i$, for all $i$, satisfies $\hat{\mathbb{E}}_t^{T_i}[\xi_i]=-\hat{\mathbb{E}}_t^{T_i}[-\xi_i]$ for $t\leq T_0$, then it holds
\begin{align*}
M_t\hat{\mathbb{E}}_t[\tilde{X}]=\sum_{i=0}^NP_t(T_i)\hat{\mathbb{E}}_t^{T_i}[\xi_i]=-M_t\hat{\mathbb{E}}_t[-\tilde{X}].
\end{align*}
\end{lem}
\begin{proof}
We derive an upper, respectively lower, bound for the upper, respectively lower, expectation of $\tilde{X}$ and show that they coincide. If we use the sublinearity of $\hat{\mathbb{E}}$ and Proposition \ref{preliminary results} $(i)$, for $t\leq T_0$, we get
\begin{align*}
M_t\hat{\mathbb{E}}_t[\tilde{X}]\leq\sum_{i=0}^NP_t(T_i)\hat{\mathbb{E}}_t^{T_i}[\xi_i],\quad-M_t\hat{\mathbb{E}}_t[-\tilde{X}]\geq\sum_{i=0}^N-P_t(T_i)\hat{\mathbb{E}}_t^{T_i}[-\xi_i].
\end{align*}
Moreover, for $t\leq T_0$, it holds $\hat{\mathbb{E}}_t[\tilde{X}]\geq-\hat{\mathbb{E}}_t[-\tilde{X}]$ and $\hat{\mathbb{E}}_t^{T_i}[\xi_i]=-\hat{\mathbb{E}}_t^{T_i}[-\xi_i]$ for all $i$, which yields the assertion.
\end{proof}
\par In general (so in particular, for contracts with asymmetric cashflows), we can use a backward induction procedure to price the contract. Then we obtain the upper and the lower expectation of the discounted payoff by recursively evaluating the forward sublinear expectation of the cashflows starting from the last cashflow.
\begin{lem}\label{pricing contracts with asymmetric cashflows}
It holds $\hat{\mathbb{E}}[\tilde{X}]=\tilde{Y}_0^+$ and $-\hat{\mathbb{E}}[-\tilde{X}]=-\tilde{Y}_0^-$, where $\tilde{Y}_i^\pm$ is defined by
\begin{align}\label{definition of Ytilde}
\tilde{Y}_i^\pm:=P_{T_{i-1}}(T_i)\hat{\mathbb{E}}_{T_{i-1}}^{T_i}[\pm\xi_i+\tilde{Y}_{i+1}^\pm]
\end{align}
for all $i=0,1,...,N$ and $T_{-1}:=0$ and $\tilde{Y}_{N+1}^\pm:=0$.
\end{lem}
\begin{proof}
We show the assertion by repeatedly excluding the cashflows from $\tilde{X}$ and using the time consistency of the $G$-expectation. First, we exclude the last cashflow from the sum and write it in terms of $\tilde{Y}_N^+$. Due to the time consistency of $\hat{\mathbb{E}}$, we have
\begin{align*}
\hat{\mathbb{E}}[\pm\tilde{X}]=\hat{\mathbb{E}}\biggl[\pm\sum_{i=0}^{N-1}M_{T_i}^{-1}\xi_i+\hat{\mathbb{E}}_{T_{N-1}}[\pm M_{T_N}^{-1}\xi_N]\biggr].
\end{align*}
By Proposition \ref{preliminary results} $(i)$, we obtain
\begin{align*}
\hat{\mathbb{E}}_{T_{N-1}}[\pm M_{T_N}^{-1}\xi_N]=M_{T_{N-1}}^{-1}P_{T_{N-1}}(T_N)\hat{\mathbb{E}}_{T_{N-1}}^{T_N}[\pm\xi_N]=M_{T_{N-1}}^{-1}\tilde{Y}_N^\pm.
\end{align*}
Second, we exclude the second last cashflow from the sum and repeat the calculation from above. Using the time consistency of $\hat{\mathbb{E}}$, we get
\begin{align*}
\hat{\mathbb{E}}[\pm\tilde{X}]=\hat{\mathbb{E}}\biggl[\pm\sum_{i=0}^{N-2}M_{T_i}^{-1}\xi_i+\hat{\mathbb{E}}_{T_{N-2}}[M_{T_{N-1}}^{-1}(\pm\xi_{N-1}+\tilde{Y}_N^\pm)]\biggr].
\end{align*}
Due to Proposition \ref{preliminary results} $(i)$, we have
\begin{align*}
\hat{\mathbb{E}}_{T_{N-2}}[M_{T_{N-1}}^{-1}(\pm\xi_{N-1}+\tilde{Y}_N^\pm)]=M_{T_{N-2}}^{-1}P_{T_{N-2}}(T_{N-1})\hat{\mathbb{E}}_{T_{N-2}}^{T_{N-1}}[\pm\xi_{N-1}+\tilde{Y}_N^\pm]=M_{T_{N-2}}^{-1}\tilde{Y}_{N-1}^\pm.
\end{align*}
Then we repeat the step from above to eventually obtain $\hat{\mathbb{E}}[\pm\tilde{X}]=\tilde{Y}_0^\pm$.
\end{proof}
\par Next, we consider a stream of options on forward prices. Most nonlinear contracts in fixed income markets can be written as a stream of options on forward prices. However, such contracts are not directly of this form but can be written to be of such a form (see Subsections \ref{caps and floors} and \ref{in-arrears contracts}). Hence, instead of specifying the payoffs in \eqref{definition of Ytilde}, we consider a slightly different sequence of random variables: for $m,n\in\mathbb{N}$ such that $m\neq n$, let $\bar{Y}_i$ be defined by
\begin{align}\label{definition of Ybar}
\bar{Y}_i:=X_{t_i}^{t_{i-1+n},t_{i+n}}\hat{\mathbb{E}}_{t_i}^{t_{i+n}}[\varphi_i(X_{t_{i+1}}^{t_{i+n},t_{i+m}})+\bar{Y}_{i+1}]
\end{align}
for all $i=1,...,N$, where $\varphi_i:\mathbb{R}\rightarrow\mathbb{R}$ and $0=t_1<...<t_{N+(m\vee n)}\leq\bar{T}$, and $\bar{Y}_{N+1}:=0$.
\par The price of such a contract is determined by a system of nonlinear PDEs. We can show that the backward induction procedure to find the price reduces to recursively solving nonlinear PDEs.
\begin{prp}\label{pricing a stream of bond options}
Let $\bar{Y}_i$ be given by \eqref{definition of Ybar} for $i=1,...,N+1$. If $\varphi_i$ satisfies \eqref{estimate} for all $i=1,...,N$, then
\begin{align*}
\bar{Y}_1=X_0^{t_n,t_{n+1}}u_1\big(0,X_0^{t_{1+n},t_{1+m}},(X_0^{t_{k-1+n},t_{k+n}},X_0^{t_{k+n},t_{k+m}})_{k=2}^N\big),
\end{align*}
where $u_i:[0,t_{i+1}]\times\mathbb{R}^{2(N-i)+1}\rightarrow\mathbb{R}$ is the unique viscosity solution to the nonlinear PDE
\begin{align}\label{recursive pde}
\begin{split}
\partial_tu+\tfrac{1}{2}\sum_{j=1}^d\Big(\overline{\sigma}_j^2\big(H_i^j(t,x_i,D_{x_i}u,D_{x_ix_i}^2u)\big)^+-\underline{\sigma}_j^2\big(H_i^j(t,x_i,D_{x_i}u,D_{x_ix_i}^2u)\big)^-\Big)={}&0,
\\u(t_{i+1},x_i)={}&f_i(x_i)
\end{split}
\end{align}
for $i=1,...,N$, where $x_i:=(\hat{x}_i,(\tilde{x}_k,\hat{x}_k)_{k=i+1}^N)$ for $i=1,...,N-1$ and $x_N:=\hat{x}_N$ and
\begin{align*}
H_i^j(t,x_i,D_{x_i}u,D_{x_ix_i}^2u):={}&\sigma_i^j(t,x_i)'D_{x_ix_i}^2u\,\sigma_i^j(t,x_i)+2D_{x_i}u\,\mu_i^j(t,x_i),
\\\sigma_i^j(t,x_i):={}&\diag(x_i)\Big(\sigma_t^j(t_{i+n},t_{i+m}),\big(\sigma_t^j(t_{k-1+n},t_{k+n}),\sigma_t^j(t_{k+n},t_{k+m})\big)_{k=i+1}^N\Big)',
\\\mu_i^j(t,x_i):={}&\diag\big(\sigma_i^j(t,x_i)\big)\Big(0,\big(\sigma_t^j(t_{k-1+n},t_{i+n}),\sigma_t^j(t_{k+n},t_{i+n})\big)_{k=i+1}^N\Big)',
\\f_i(x_i):={}&\varphi_i(\hat{x}_i)+\tilde{x}_{i+1}u_{i+1}(t_{i+1},x_{i+1})
\end{align*}
for $i=1,...,N-1$ and
\begin{align*}
H_N^j(t,x_N,D_{x_N}u,D_{x_Nx_N}^2u):={}&\sigma_t^j(t_{N+n},t_{N+m})^2x_N^2\partial_{x_Nx_N}^2u,
\\f_N(x_N):={}&\varphi_N(x_N).
\end{align*}
\end{prp}
\begin{proof}
We apply the nonlinear Feynman-Kac formula of \citet*{hujipengsong2014} to $\bar{Y}_i$ for all $i=1,...,N$ to show that one can recursively solve the nonlinear PDE \eqref{recursive pde} to obtain $\bar{Y}_1$.
\par We start with $\bar{Y}_N$. By Proposition \ref{pricing a bond option}, we know that
\begin{align*}
\bar{Y}_N=X_{t_N}^{t_{N-1+n},t_{N+n}}u_N(t_N,X_{t_N}^{t_{N+n},t_{N+m}}),
\end{align*}
where $u_N(t_N,\cdot)$ satisfies \eqref{estimate} \citep*[Proposition 4.2]{hujipengsong2014}.
\par Now we move on to $\bar{Y}_{N-1}$. Inserting $\bar{Y}_N$ in the definition of $\bar{Y}_{N-1}$, we get
\begin{align*}
\bar{Y}_{N-1}=X_{t_{N-1}}^{t_{N-2+n},t_{N-1+n}}\hat{\mathbb{E}}_{t_{N-1}}^{t_{N-1+n}}[f_{N-1}(X_{t_N}^{t_{N-1+n},t_{N-1+m}},X_{t_N}^{t_{N-1+n},t_{N+n}},X_{t_N}^{t_{N+n},t_{N+m}})].
\end{align*}
One can show that $f_{N-1}$ satisfies \eqref{estimate}, since $\varphi_{N-1}$ and $u_N(t_N,\cdot)$ satisfy \eqref{estimate}. Hence, we can apply the nonlinear Feynman-Kac formula---as in the proof of Proposition \ref{pricing a bond option}---to obtain
\begin{align*}
\bar{Y}_{N-1}=X_{t_{N-1}}^{t_{N-2+n},t_{N-1+n}}u_{N-1}(t_{N-1},X_{t_{N-1}}^{t_{N-1+n},t_{N-1+m}},X_{t_{N-1}}^{t_{N-1+n},t_{N+n}},X_{t_{N-1}}^{t_{N+n},t_{N+m}}),
\end{align*}
where $u_{N-1}(t_{N-1},\cdot)$ satisfies \eqref{estimate} \citep*[Proposition 4.2]{hujipengsong2014}.
\par Next, we perform the following recursive step for all $i=1,...,N-2$ backwards to obtain $\bar{Y}_1$. Let us suppose that
\begin{align*}
\bar{Y}_{i+1}=X_{t_{i+1}}^{t_{i+n},t_{i+1+n}}u_{i+1}\big(t_{i+1},X_{t_{i+1}}^{t_{i+1+n},t_{i+1+m}}(X_{t_{i+1}}^{t_{k-1+n},t_{k+n}},X_{t_{i+1}}^{t_{k+n},t_{k+m}})_{k=i+2}^N\big),
\end{align*}
where $u_{i+1}(t_{i+1},\cdot)$ satisfies \eqref{estimate}. Plugging $\bar{Y}_{i+1}$ into the definition of $\bar{Y}_i$ yields
\begin{align*}
\bar{Y}_i=X_{t_i}^{t_{i-1+n},t_{i+n}}\hat{\mathbb{E}}_{t_i}^{t_{i+n}}\big[f_i\big(X_{t_{i+1}}^{t_{i+n},t_{i+m}},(X_{t_{i+1}}^{t_{k-1+n},t_{k+n}},X_{t_{i+1}}^{t_{k+n},t_{k+m}})_{k=i+1}^N\big)\big].
\end{align*}
As in the previous step, one can show that $f_i$ satisfies \eqref{estimate}. Therefore, the nonlinear Feynman-Kac formula implies
\begin{align*}
\bar{Y}_i=X_{t_i}^{t_{i-1+n},t_{i+n}}u_i\big(t_i,X_{t_i}^{t_{i+n},t_{i+m}},(X_{t_i}^{t_{k-1+n},t_{k+n}},X_{t_i}^{t_{k+n},t_{k+m}})_{k=i+1}^N\big),
\end{align*}
where $u_i(t_i,\cdot)$ satisfies \eqref{estimate} \citep*[Proposition 4.2]{hujipengsong2014}.
\end{proof}
\par When the contract consists of options that are additionally convex or concave, the price is determined by the price from the corresponding HJM model without volatility uncertainty. If all payoff functions are convex, respectively concave, then the backward induction procedure reduces to computing the linear expectation of all cashflows when the forward price dynamics are driven by a standard Brownian motion with volatility $\overline{\sigma}$, respectively $\underline{\sigma}$.
\begin{prp}\label{pricing a stream of convex bond options}
Let $\bar{Y}_i$ be given by \eqref{definition of Ybar} for $i=1,...,N+1$. If $\varphi_i$ is convex and satisfies \eqref{estimate} for all $i=1,...,N$, then
\begin{align*}
\bar{Y}_1=\sum_{i=1}^NX_0^{t_n,t_{i+n}}u_i^{\overline{\sigma}}(0,X_0^{t_{i+n},t_{i+m}}),
\end{align*}
where the function $u_i^\sigma:[0,t_{i+1}]\times\mathbb{R}\rightarrow\mathbb{R}$, for all $i=1,...,N$ and $\sigma\in\Sigma$, is defined by
\begin{align*}
u_i^\sigma(t,\hat{x}_i):=\mathbb{E}_{P_0}[\varphi_i(X_{t_{i+1}}^i)]
\end{align*}
and the process $X^i=(X_s^i)_{t\leq s\leq t_{i+1}}$ is given by
\begin{align*}
X_s^i=\hat{x}_i-\sum_{j=1}^d\int_t^s\sigma_u^j(t_{i+n},t_{i+m})X_u^i\sigma_jdB_u^j.
\end{align*}
If $\varphi_i$ is concave instead of convex for all $i=1,...,N$, then
\begin{align*}
\bar{Y}_1=\sum_{i=1}^NX_0^{t_n,t_{i+n}}u_i^{\underline{\sigma}}(0,X_0^{t_{i+n},t_{i+m}}).
\end{align*}
\end{prp}
\begin{proof}
We solve the nonlinear PDE \eqref{recursive pde} for all $i=1,...,N$ and use Proposition \ref{pricing a stream of bond options} to prove the first assertion; the proof of the second assertion is similar. Moreover, we only consider the case in which $\tilde{x}_i\geq0$ for all $i=2,...,N$---this is sufficient as the forward prices are positive.
\par First of all, we can show that
\begin{align*}
u_N(t,x_N)=u_N^{\overline{\sigma}}(t,\hat{x}_N),
\end{align*}
since $u_i^{\overline{\sigma}}(t,\cdot)$ is convex and $u_i^{\overline{\sigma}}$ is the solution to
\begin{align*}
\partial_tu+\tfrac{1}{2}\sum_{j=1}^d\overline{\sigma}_j^2\sigma_t^j(t_{i+n},t_{i+m})^2\hat{x}_i^2\partial_{\hat{x}_i\hat{x}_i}^2u=0,\quad u(t_{i+1},\hat{x}_i)={}&\varphi_i(\hat{x}_i)
\end{align*}
for all $i=1,...,N$.
\par Second, we show by verification that
\begin{align*}
u_{N-1}(t,x_{N-1})=u_{N-1}^{\overline{\sigma}}(t,\hat{x}_{N-1})+\tilde{x}_Nu_N^{\overline{\sigma}}(t,\hat{x}_N).
\end{align*}
Using the previous equation and performing some calculations leads to
\begin{align*}
H_{N-1}^j(t,x_{N-1},D_{x_{N-1}}u_{N-1},D_{x_{N-1}x_{N-1}}^2u_{N-1})={}&\sigma_t^j(t_{N-1+n},t_{N-1+m})^2\hat{x}_{N-1}^2\partial_{\hat{x}_{N-1}\hat{x}_{N-1}}^2u_{N-1}^{\overline{\sigma}}
\\&{}+\tilde{x}_N\sigma_t^j(t_{N+n},t_{N+m})^2\hat{x}_N^2\partial_{\hat{x}_N\hat{x}_N}^2u_N^{\overline{\sigma}}.
\end{align*}
By the arguments from the first step, we then have
\begin{align*}
H_{N-1}^j(t,x_{N-1},D_{x_{N-1}}u_{N-1},D_{x_{N-1}x_{N-1}}^2u_{N-1})\geq0,
\end{align*}
and (therefore) one can verify that $u_{N-1}$ indeed solves \eqref{recursive pde} for $i=N-1$.
\par Next, we do the following recursive step for all $i=1,...,N-2$ backwards to get an expression for $u_1$. Let us suppose that
\begin{align*}
u_{i+1}(t,x_{i+1})=u_{i+1}^{\overline{\sigma}}(t,\hat{x}_{i+1})+\tilde{x}_{i+2}u_{i+2}(t,x_{i+2})
\end{align*}
and that
\begin{align*}
H_{i+1}^j(t,x_{i+1},D_{x_{i+1}}u_{i+1},D_{x_{i+1}x_{i+1}}^2u_{i+1})\geq0.
\end{align*}
Then we show by verification that
\begin{align*}
u_i(t,x_i)=u_i^{\overline{\sigma}}(t,\hat{x}_i)+\tilde{x}_{i+1}u_{i+1}(t,x_{i+1}).
\end{align*}
If we use the above equation and do some calculations, we obtain
\begin{align*}
H_i^j(t,x_i,D_{x_i}u_i,D_{x_ix_i}^2u_i)={}&\sigma_t^j(t_{i+n},t_{i+m})^2\hat{x}_i^2\partial_{\hat{x}_i\hat{x}_i}^2u_i^{\overline{\sigma}}
\\&{}+\tilde{x}_{i+1}H_{i+1}^j(t,x_{i+1},D_{x_{i+1}}u_{i+1},D_{x_{i+1}x_{i+1}}^2u_{i+1}).
\end{align*}
As in the previous step, we then have
\begin{align*}
H_i^j(t,x_i,D_{x_i}u_i,D_{x_ix_i}^2u_i)\geq0,
\end{align*}
and (thus) one can verify that $u_i$ solves \eqref{recursive pde}. If we recursively plug in the explicit solutions, we eventually obtain
\begin{align*}
u_1=u_1^{\overline{\sigma}}(t,\hat{x}_1)+\sum_{k=2}^N\biggl(\prod_{l=2}^k\tilde{x}_l\biggr)u_k^{\overline{\sigma}}(t,\hat{x}_k).
\end{align*}
\par By Proposition \ref{pricing a stream of bond options} and the definition of the forward prices, we finally obtain the desired expression for $\bar{Y}_1$.
\end{proof}

\section{Common Interest Rate Derivatives}\label{common interest rate derivatives}
With the tools from the preceding sections, we price all major contracts traded in fixed income markets. We consider typical linear contracts, such as fixed coupon bonds, floating rate notes, and interest rate swaps, and nonlinear contracts, such as swaptions, caps and floors, and in-arrears contracts. Using the pricing techniques from Sections \ref{pricing single cashflows} and \ref{pricing a stream of cashflows}, we show how to derive robust pricing formulas for such contracts. That means, we consider a contract with discounted payoff
\begin{align*}
\tilde{X}=\sum_{i=0}^NM_{T_i}^{-1}\xi_i
\end{align*}
for $0<T_0<T_1<...<T_N=\bar{T}$ and specifically given cashflows, and then we show how to find $\hat{\mathbb{E}}[\tilde{X}]$ and $-\hat{\mathbb{E}}[-\tilde{X}]$ or $M_t\hat{\mathbb{E}}_t[\tilde{X}]$ and $-M_t\hat{\mathbb{E}}_t[-\tilde{X}]$ for $t\leq T_0$ if the contract has a symmetric payoff.

\subsection{Fixed Coupon Bonds}\label{fixed coupon bonds}
We can price fixed coupon bonds as in the classical case without volatility uncertainty. A fixed coupon bond is a contract that pays a fixed rate of interest, given by $K>0$, on a nominal value, which is normalized to $1$, at each payment date and the nominal value at the last payment date. Hence, the cashflows are given by
\begin{align}\label{fixed coupon bond cashflows}
\xi_i=1_{\{N\}}(i)+1_{\{1,...,N\}}(i)(T_i-T_{i-1})K
\end{align}
for all $i=0,1,...,N$. Due to its simple payoff structure, the contract has a symmetric payoff, and its price is given by the same expression as the one obtained in traditional term structure models.
\begin{prp}\label{fixed coupon bond price}
Let $\xi_i$ be given by \eqref{fixed coupon bond cashflows} for all $i=0,1,...,N$. Then for $t\leq T_0$,
\begin{align*}
M_t\hat{\mathbb{E}}_t[\tilde{X}]=P_t(T_N)+\sum_{i=1}^NP_t(T_i)(T_i-T_{i-1})K=-M_t\hat{\mathbb{E}}_t[-\tilde{X}].
\end{align*}
\end{prp}
\begin{proof}
Since the cashflows are constants, the assertion follows by Lemma \ref{pricing contracts with symmetric cashflows}.
\end{proof}

\subsection{Floating Rate Notes}\label{floating rate notes}
We can also price floating rate notes as in the classical case without volatility uncertainty. A floating rate note is a fixed coupon bond in which the fixed rate is replaced by a floating rate: the simply compounded spot rate; for $t\leq T\leq\bar{T}$, the simply compounded spot rate with maturity $T$ at time $t$ is defined by
\begin{align*}
L_t(T):=\tfrac{1}{T-t}(\tfrac{1}{P_t(T)}-1).
\end{align*}
The cashflows are then given by
\begin{align}\label{floating rate note cashflows}
\xi_i=1_{\{N\}}(i)+1_{\{1,...,N\}}(i)(T_i-T_{i-1})L_{T_{i-1}}(T_i)
\end{align}
for all $i=0,1,...,N$. Although the cashflows are not constant, the contract yet has a symmetric payoff. As in the classical case, the price is simply given by the price of a zero-coupon bond with maturity $T_0$.
\begin{prp}\label{floating rate note price}
Let $\xi_i$ be given by \eqref{floating rate note cashflows} for all $i=0,1,...,N$. Then for $t\leq T_0$,
\begin{align*}
M_t\hat{\mathbb{E}}_t[\tilde{X}]=P_t(T_0)=-M_t\hat{\mathbb{E}}_t[-\tilde{X}].
\end{align*}
\end{prp}
\begin{proof}
We show that the cashflows have a symmetric payoff and apply Lemma \ref{pricing contracts with symmetric cashflows}. Due to Proposition \ref{preliminary results} $(ii)$ and $(iv)$, we have
\begin{align*}
P_t(T_i)\hat{\mathbb{E}}_t^{T_i}[(T_i-T_{i-1})L_{T_{i-1}}(T_i)]=P_t(T_{i-1})\hat{\mathbb{E}}_t^{T_{i-1}}[1-P_{T_{i-1}}(T_i)]=P_t(T_{i-1})-P_t(T_i)
\end{align*}
for all $i=1,...,N$. In a similar fashion we can show that
\begin{align*}
-P_t(T_i)\hat{\mathbb{E}}_t^{T_i}[-(T_i-T_{i-1})L_{T_{i-1}}(T_i)]=P_t(T_{i-1})-P_t(T_i)
\end{align*}
for all $i=1,...,N$. The result follows by Lemma \ref{pricing contracts with symmetric cashflows} and summation.
\end{proof}

\subsection{Interest Rate Swaps}\label{interest rate swaps}
The pricing formula for interest rate swaps is the same as in traditional models. An interest rate swap exchanges the floating rate with a fixed rate at each payment date. Without loss of generality, we consider a payer interest rate swap; that is, we pay the fixed rate and receive the floating rate. Hence, the cashflows are given by
\begin{align}\label{interest rate swap cashflows}
\xi_i=1_{\{1,...,N\}}(i)(T_i-T_{i-1})\big(L_{T_{i-1}}(T_i)-K\big)
\end{align}
for all $i=0,1,...,N$. Since the payoff is the difference of a zero-coupon bond and a floating rate note, the contract is symmetric. As in traditional term structure models, the price is given by a linear combination of zero-coupon bonds with different maturities. In particular, this implies that the swap rate---i.e., the value of the fixed rate that makes the value of the contract zero---is uniquely determined and does not differ from the expression obtained by standard models.
\begin{prp}\label{interest rate swap price}
Let $\xi_i$ be given by \eqref{interest rate swap cashflows} for all $i=0,1,...,N$. Then for $t\leq T_0$,
\begin{align*}
M_t\hat{\mathbb{E}}_t[\tilde{X}]=P_t(T_0)-P_t(T_N)-\sum_{i=1}^NP_t(T_i)(T_i-T_{i-1})K=-M_t\hat{\mathbb{E}}_t[-\tilde{X}].
\end{align*}
\end{prp}
\begin{proof}
Again, we show that the cashflows have a symmetric payoff and use Lemma \ref{pricing contracts with symmetric cashflows} to obtain the result. As in the proof of Proposition \ref{floating rate note price}, we can show that
\begin{align*}
P_t(T_i)\hat{\mathbb{E}}_t^{T_i}\big[(T_i-T_{i-1})\big(L_{T_{i-1}}(T_i)-K\big)\big]={}&P_t(T_{i-1})-P_t(T_i)-P_t(T_i)(T_i-T_{i-1})K,
\\-P_t(T_i)\hat{\mathbb{E}}_t^{T_i}\big[-(T_i-T_{i-1})\big(L_{T_{i-1}}(T_i)-K\big)\big]={}&P_t(T_{i-1})-P_t(T_i)-P_t(T_i)(T_i-T_{i-1})K
\end{align*}
for all $i=1,...,N$. Then the assertion follows by Lemma \ref{pricing contracts with symmetric cashflows} and summation.
\end{proof}

\subsection{Swaptions}\label{swaptions}
We can price swaptions by computing the price in the corresponding HJM model without volatility uncertainty to obtain the upper and the lower bound for the price. A swaption gives the buyer the right to enter an interest rate swap at the first payment date. Hence, there is only one cashflow, which is determined by Proposition \ref{interest rate swap price}---i.e.,
\begin{align}\label{swaption cashflows}
\xi_i=1_{\{0\}}(i)\biggl(1-P_{T_0}(T_n)-\sum_{j=1}^NP_{T_0}(T_j)(T_j-T_{j-1})K\biggr)^+
\end{align}
for all $i=0,1,...,N$. Due to the nonlinearity of the payoff function, the upper and the lower expectation of the discounted payoff do not necessarily coincide; thus, the contract has an asymmetric payoff. The related pricing bounds are given by the linear expectation of the payoff when the forward prices are driven by a standard Brownian motion with the highest and the lowest possible volatility, respectively.
\begin{thm}\label{swaption price}
Let $\xi_i$ be given by \eqref{swaption cashflows} for all $i=0,1,...,N$. Then it holds
\begin{align*}
\hat{\mathbb{E}}[\tilde{X}]=P_0(T_0)u^{\overline{\sigma}}\Big(0,\big(\tfrac{P_0(T_i)}{P_0(T_0)}\big)_{i=1}^N\Big),\quad-\hat{\mathbb{E}}[-\tilde{X}]=P_0(T_0)u^{\underline{\sigma}}\Big(0,\big(\tfrac{P_0(T_i)}{P_0(T_0)}\big)_{i=1}^N\Big),
\end{align*}
where the function $u^\sigma:[0,T_0]\times\mathbb{R}^N\rightarrow\mathbb{R}$, for $\sigma\in\Sigma$, is defined by
\begin{align*}
u^\sigma(t,x):=\mathbb{E}_{P_0}\biggl[\biggl(1-X_{T_0}^N-\sum_{i=1}^NX_{T_0}^i(T_i-T_{i-1})K\biggr)^+\biggr]
\end{align*}
and the process $X^i=(X_s^i)_{t\leq s\leq T_0}$, for all $i=1,...,N$, is given by
\begin{align*}
X_s^i=x_i-\sum_{j=1}^d\int_t^s\sigma_u^j(T_0,T_i)X_u^i\sigma_jdB_u^j.
\end{align*}
\end{thm}
\begin{proof}
We prove the claim by using Proposition \ref{pricing a convex bond option}. By Proposition \ref{preliminary results} $(i)$, we have
\begin{align*}
\hat{\mathbb{E}}[\tilde{X}]={}&P_0(T_0)\hat{\mathbb{E}}^{T_0}\biggl[\biggl(1-X_{T_0}^{T_0,T_N}-\sum_{i=1}^NX_{T_0}^{T_0,T_i}(T_i-T_{i-1})K\biggr)^+\biggr],
\\-\hat{\mathbb{E}}[-\tilde{X}]={}&-P_0(T_0)\hat{\mathbb{E}}^{T_0}\biggl[-\biggl(1-X_{T_0}^{T_0,T_N}-\sum_{i=1}^NX_{T_0}^{T_0,T_i}(T_i-T_{i-1})K\biggr)^+\biggr].
\end{align*}
Hence, the assertion follows by Proposition \ref{pricing a convex bond option}, since one can show that the payoff function of a swaption is convex and satisfies \eqref{estimate}.
\end{proof}

\subsection{Caps and Floors}\label{caps and floors}
Similar to swaptions, we can compute the upper and the lower bound for the price of a cap by pricing it in the corresponding HJM model without volatility uncertainty. A cap gives the buyer the right to exchange the floating rate with a fixed rate at each payment date. The cashflows are called caplets and are given by
\begin{align}\label{cap cashflows}
\xi_i=1_{\{1,...,N\}}(i)(T_i-T_{i-1})\big(L_{T_{i-1}}(T_i)-K\big)^+
\end{align}
for all $i=0,1,...,N$. The upper and the lower bound for the price of the contract are given by the linear expectation of its payoff, which corresponds to a collection of put options on forward prices, when the forward prices are driven by a standard Brownian motion with the highest and the lowest possible volatility, respectively.
\begin{thm}\label{cap price}
Let $\xi_i$ be given by \eqref{cap cashflows} for all $i=0,1,...,N$. Then it holds
\begin{align*}
\hat{\mathbb{E}}[\tilde{X}]=\sum_{i=1}^NP_0(T_{i-1})u_i^{\overline{\sigma}}\big(0,\tfrac{P_0(T_i)}{P_0(T_{i-1})}\big),\quad-\hat{\mathbb{E}}[-\tilde{X}]=\sum_{i=1}^NP_0(T_{i-1})u_i^{\underline{\sigma}}\big(0,\tfrac{P_0(T_i)}{P_0(T_{i-1})}\big),
\end{align*}
where the function $u_i^\sigma:[0,T_{i-1}]\times\mathbb{R}\rightarrow\mathbb{R}$, for all $i=1,...,N$ and $\sigma\in\Sigma$, is defined by
\begin{align*}
u_i^\sigma(t,x_i):=\tfrac{1}{K_i}\mathbb{E}_{P_0}[(K_i-X_{T_{i-1}}^i)^+]
\end{align*}
for $K_i:=\frac{1}{1+(T_i-T_{i-1})K}$ and the process $X^i=(X_s^i)_{t\leq s\leq T_{i-1}}$ is given by
\begin{align*}
X_s^i=x_i-\sum_{j=1}^d\int_t^s\sigma_u^j(T_{i-1},T_i)X_u^i\sigma_jdB_u^j.
\end{align*}
\end{thm}
\begin{proof}
We use Lemma \ref{pricing contracts with asymmetric cashflows} and Proposition \ref{pricing a stream of convex bond options} to show the assertion. According to Lemma \ref{pricing contracts with asymmetric cashflows}, we need to determine $\tilde{Y}_0^\pm$ in order to obtain $\hat{\mathbb{E}}[\pm\tilde{X}]$. We compute $\tilde{Y}_0^\pm$ by using Proposition \ref{pricing a stream of convex bond options}. For this purpose, we need to rewrite $\tilde{Y}_i^\pm$ for all $i=0,1,...,N$ and define a sequence of random variables to which we can apply Proposition \ref{pricing a stream of convex bond options}. For all $i=0,1,...,N$, we have
\begin{align*}
\tilde{Y}_i^\pm=P_{T_{i-1}}(T_i)\hat{\mathbb{E}}_{T_{i-1}}^{T_i}[\pm\xi_i+\tilde{Y}_{i+1}^\pm],
\end{align*}
where $\xi_i$ is given by \eqref{cap cashflows}, and $\tilde{Y}_{N+1}^\pm=0$. Since $\xi_i\in L_G^1(\Omega_{T_{i-1}})$ for all $i=1,...,N$ and $\xi_0=0$, we can show that
\begin{align*}
\tilde{Y}_i^\pm=\pm\tfrac{1}{K_i}(K_i-X_{T_{i-1}}^{T_{i-1},T_i})^++X_{T_{i-1}}^{T_{i-1},T_i}\hat{\mathbb{E}}_{T_{i-1}}^{T_i}[\tilde{Y}_{i+1}^\pm]
\end{align*}
for all $i=1,...,N$ and $\tilde{Y}_0^\pm=X_0^{0,T_0}\hat{\mathbb{E}}^{T_0}[\tilde{Y}_1^\pm]$. Now we define $\bar{Y}_i^\pm:=X_{T_{i-2}}^{T_{i-2},T_{i-1}}\hat{\mathbb{E}}_{T_{i-2}}^{T_{i-1}}[\tilde{Y}_i^\pm]$ for all $i=1,...,N+1$. Then we have $\tilde{Y}_0^\pm=\bar{Y}_1^\pm$ and
\begin{align*}
\bar{Y}_i^\pm=X_{T_{i-2}}^{T_{i-2},T_{i-1}}\hat{\mathbb{E}}_{T_{i-2}}^{T_{i-1}}[\pm\tfrac{1}{K_i}(K_i-X_{T_{i-1}}^{T_{i-1},T_i})^++\bar{Y}_{i+1}^\pm]
\end{align*}
for all $i=1,...,N$, where $\bar{Y}_{N+1}^\pm=0$. Moreover, we define $t_i:=T_{i-2}$ for all $i=1,...,N+2$. Then it holds $0=t_1<...<t_{N+2}\leq\bar{T}$ and
\begin{align*}
\bar{Y}_i^\pm=X_{t_i}^{t_i,t_{i+1}}\hat{\mathbb{E}}_{t_i}^{t_{i+1}}[\pm\tfrac{1}{K_i}(K_i-X_{t_{i+1}}^{t_{i+1},t_{i+2}})^++\bar{Y}_{i+1}^\pm]
\end{align*}
for all $i=1,...,N$. Thus, we can apply Proposition \ref{pricing a stream of convex bond options} to obtain
\begin{align*}
\bar{Y}_1^+=\sum_{i=1}^NX_0^{0,t_{i+1}}u_i^{\overline{\sigma}}(0,X_0^{t_{i+1},t_{i+2}}),\quad\bar{Y}_1^-=\sum_{i=1}^N-X_0^{0,t_{i+1}}u_i^{\underline{\sigma}}(0,X_0^{t_{i+1},t_{i+2}}),
\end{align*}
which proves the assertion.
\end{proof}
\par Floors can be priced in the same manner as caps. A floor gives the buyer the right to exchange a fixed rate with the floating rate at each payment date. The cashflows are called floorlets and are given by
\begin{align}\label{floor cashflows}
\xi_i=1_{\{1,...,N\}}(i)(T_i-T_{i-1})\big(K-L_{T_{i-1}}(T_i)\big)^+
\end{align}
for all $i=0,1,...,N$. Since the cashflows are very similar to caplets, we obtain similar pricing bounds compared to Theorem \ref{cap price}; the only difference is that we need to compute prices of call options on forward prices instead of put options to obtain the pricing bounds. It is remarkable that we can show this with the put-call parity, since the nonlinearity of the pricing measure implies that the put-call parity, in general, does not hold in the presence of volatility uncertainty.
\begin{thm}\label{floor price}
Let $\xi_i$ be given by \eqref{floor cashflows} for all $i=0,1,...,N$. Then it holds
\begin{align*}
\hat{\mathbb{E}}[\tilde{X}]=\sum_{i=1}^NP_0(T_{i-1})u_i^{\overline{\sigma}}\big(0,\tfrac{P_0(T_i)}{P_0(T_{i-1})}\big),\quad-\hat{\mathbb{E}}[-\tilde{X}]=\sum_{i=1}^NP_0(T_{i-1})u_i^{\underline{\sigma}}\big(0,\tfrac{P_0(T_i)}{P_0(T_{i-1})}\big),
\end{align*}
where the function $u_i^\sigma:[0,T_{i-1}]\times\mathbb{R}\rightarrow\mathbb{R}$, for all $i=1,...,N$ and $\sigma\in\Sigma$, is defined by
\begin{align*}
u_i^\sigma(t,x_i):=\tfrac{1}{K_i}\mathbb{E}_P[(X_{T_{i-1}}^i-K_i)^+]
\end{align*}
and $K_i$ and the process $X^i=(X_s^i)_{t\leq s\leq T_{i-1}}$ are given as in Theorem \ref{cap price}.
\end{thm}
\begin{proof}
Although $\hat{\mathbb{E}}$ is sublinear, we can still use the put-call parity to prove the claim, since interest rate swaps have a symmetric payoff. For all $i=1,...,N$, we have
\begin{align*}
\xi_i=(T_i-T_{i-1})\big(L_{T_{i-1}}(T_i)-K\big)^+-(T_i-T_{i-1})\big(L_{T_{i-1}}(T_i)-K\big).
\end{align*}
Thus, we get $\tilde{X}=\tilde{Y}-\tilde{Z}$, where $\tilde{Y}$, respectively $\tilde{Z}$, denotes the discounted payoff of a cap, respectively interest rate swap; that is,
\begin{align*}
\tilde{Y}:=\sum_{i=1}^NM_{T_i}^{-1}(T_i-T_{i-1})\big(L_{T_{i-1}}(T_i)-K\big)^+,
\quad\tilde{Z}:=\sum_{i=1}^NM_{T_i}^{-1}(T_i-T_{i-1})\big(L_{T_{i-1}}(T_i)-K\big).
\end{align*}
Due to the sublinearity of $\hat{\mathbb{E}}$, we get $\hat{\mathbb{E}}[\tilde{X}]\leq\hat{\mathbb{E}}[\tilde{Y}]+\hat{\mathbb{E}}[-\tilde{Z}]$ and $\hat{\mathbb{E}}[\tilde{X}]\geq\hat{\mathbb{E}}[\tilde{Y}]-\hat{\mathbb{E}}[\tilde{Z}]$. Hence, by Proposition \ref{interest rate swap price}, we obtain $\hat{\mathbb{E}}[\tilde{X}]=\hat{\mathbb{E}}[\tilde{Y}]-\hat{\mathbb{E}}[\tilde{Z}]$. In a similar fashion, we can show that $-\hat{\mathbb{E}}[-\tilde{X}]=-\hat{\mathbb{E}}[-\tilde{Y}]-\hat{\mathbb{E}}[\tilde{Z}]$. Therefore, the assertion follows by the classical put-call parity.
\end{proof}

\subsection{In-Arrears Contracts}\label{in-arrears contracts}
The pricing procedure from the previous subsection also works for contracts in which the floating rate is settled in arrears. The difference between the contracts from above and in-arrears contracts is that the simply compounded spot rate is reset each time when the contract pays off. As a representative contract, we show how to price in-arrears swaps; other contracts, such as in-arrears caps and floors, can be priced in a similar way. In contrast to a plain vanilla interest rate swap, the cashflows are now given by
\begin{align}\label{in-arrears swap cashflows}
\xi_i=1_{\{0,1,...,N-1\}}(i)(T_{i+1}-T_i)\big(L_{T_i}(T_{i+1})-K\big)
\end{align}
for all $i=0,1,...,N$. Then the contract is not necessarily symmetric, and the pricing bounds are given by the linear expectation of its payoff, corresponding to a collection of functions depending on forward prices, when the forward prices are driven by a standard Brownian motion with the highest and the lowest possible volatility, respectively. As a consequence, there is not a unique swap rate for in-arrears swaps.
\begin{thm}\label{in-arrears swap price}
Let $\xi_i$ be given by \eqref{in-arrears swap cashflows} for all $i=0,1,...,N$. Then it holds
\begin{align*}
\hat{\mathbb{E}}[\tilde{X}]=\sum_{i=1}^NP_0(T_i)u_i^{\overline{\sigma}}\big(0,\tfrac{P_0(T_{i-1})}{P_0(T_i)}\big),\quad-\hat{\mathbb{E}}[-\tilde{X}]=\sum_{i=1}^NP_0(T_i)u_i^{\underline{\sigma}}\big(0,\tfrac{P_0(T_{i-1})}{P_0(T_i)}\big),
\end{align*}
where the function $u_i^\sigma:[0,T_{i-1}]\times\mathbb{R}\rightarrow\mathbb{R}$, for all $i=1,...,N$ and $\sigma\in\Sigma$, is defined by
\begin{align*}
u_i^\sigma(t,x_i):=\mathbb{E}_{P_0}[X_{T_{i-1}}^i(X_{T_{i-1}}^i-\tfrac{1}{K_i})]
\end{align*}
for $K_i$ as in Theorem \ref{cap price} and the process $X^i=(X_s^i)_{t\leq s\leq T_{i-1}}$ is given by
\begin{align*}
X_s^i=x_i-\sum_{j=1}^d\int_t^s\sigma_u^j(T_i,T_{i-1})X_u^i\sigma_jdB_u^j.
\end{align*}
\end{thm}
\begin{proof}
As in the proof of Theorem \ref{cap price}, we use Lemma \ref{pricing contracts with asymmetric cashflows} and Proposition \ref{pricing a stream of convex bond options} to show the assertion. By Lemma \ref{pricing contracts with asymmetric cashflows}, we need to compute $\tilde{Y}_0^\pm$ to find $\hat{\mathbb{E}}[\pm\tilde{X}]$. In order to find $\tilde{Y}_0^\pm$, we rewrite $\tilde{Y}_i^\pm$ for all $i=0,1,...,N$ and define a sequence of random variables to which we can apply Proposition \ref{pricing a stream of convex bond options}. For all $i=0,1,...,N$, we have
\begin{align*}
\tilde{Y}_i^\pm=P_{T_{i-1}}(T_i)\hat{\mathbb{E}}_{T_{i-1}}^{T_i}[\pm\xi_i+\tilde{Y}_{i+1}^\pm],
\end{align*}
where $\xi_i$ is given by \eqref{in-arrears swap cashflows} and $\tilde{Y}_{N+1}^\pm=0$. Since $\xi_N=0$, we get $\tilde{Y}_N^\pm=0$. For all $i=0,1,...,N-1$, we obtain, by Proposition \ref{preliminary results} $(ii)$,
\begin{align*}
\tilde{Y}_i^\pm=X_{T_{i-1}}^{T_{i-1},T_{i+1}}\hat{\mathbb{E}}_{T_{i-1}}^{T_{i+1}}[\pm X_{T_i}^{T_{i+1},T_i}(X_{T_i}^{T_{i+1},T_i}-\tfrac{1}{K_{i+1}})+X_{T_i}^{T_{i+1},T_i}\tilde{Y}_{i+1}^\pm].
\end{align*}
We define $\bar{Y}_i^\pm:=X_{T_{i-2}}^{T_{i-1},T_{i-2}}\tilde{Y}_{i-1}^\pm$ for all $i=1,...,N+1$. Then it holds $\tilde{Y}_0^\pm=X_0^{0,T_0}\bar{Y}_1^\pm$ and
\begin{align*}
\bar{Y}_i^\pm=X_{T_{i-2}}^{T_{i-1},T_i}\hat{\mathbb{E}}_{T_{i-2}}^{T_i}[\pm X_{T_{i-1}}^{T_i,T_{i-1}}(X_{T_{i-1}}^{T_i,T_{i-1}}-\tfrac{1}{K_i})+\bar{Y}_{i+1}^\pm]
\end{align*}
for all $i=1,...,N$, where $\bar{Y}_{N+1}^\pm=0$. Furthermore, we set $t_i:=T_{i-2}$ for all $i=1,...,N+2$. Then we get $0=t_1<...<t_{N+2}\leq\bar{T}$ and
\begin{align*}
\bar{Y}_i^\pm=X_{t_i}^{t_{i+1},t_{i+2}}\hat{\mathbb{E}}_{t_i}^{t_{i+2}}[\pm X_{t_{i+1}}^{t_{i+2},t_{i+1}}(X_{t_{i+1}}^{t_{i+2},t_{i+1}}-\tfrac{1}{K_i})+\bar{Y}_{i+1}^\pm]
\end{align*}
for all $i=1,...,N$. Therefore, by Proposition \ref{pricing a stream of convex bond options}, it holds
\begin{align*}
\bar{Y}_1^+=\sum_{i=1}^NX_0^{t_2,t_{i+2}}u_i^{\overline{\sigma}}(0,X_0^{t_{i+2},t_{i+1}}),\quad\bar{Y}_1^-=\sum_{i=1}^N-X_0^{t_2,t_{i+2}}u_i^{\underline{\sigma}}(0,X_0^{t_{i+2},t_{i+1}}),
\end{align*}
which proves the assertion.
\end{proof}

\subsection{Other Contracts}\label{other contracts}
The pricing of other (more complex) contracts requires numerical methods. Almost all contracts in fixed income markets correspond to (a collection of) options on forward prices, and Propositions \ref{pricing a convex bond option} and \ref{pricing a stream of convex bond options} show that---as demonstrated in the preceding subsections---we can reduce the problem of pricing contracts to pricing them in the corresponding HJM model without volatility uncertainty whenever the payoffs are convex or concave. However, this is not always the case. Therefore, we need a different pricing procedure in the remaining cases, that is, when payoffs are not convex or concave. In such cases, we can solve the nonlinear PDEs arising due to Propositions \ref{pricing a bond option} and \ref{pricing a stream of bond options} in order to determine prices of contracts. For this purpose, we can generally use all numerical schemes for solving nonlinear PDEs. In particular, there are various numerical approaches in the literature on robust finance to price derivatives written on asset prices under volatility uncertainty \citep*{avellanedalevyparas1995,guyonhenrylabordere2011,lyons1995,nendel2021}, including trinomial tree approximations of stochastic differential equations and Monte Carlo methods in addition to classical methods for solving PDEs numerically. How such approaches perform in the present setting and which approach works best are interesting questions for future research, since most fixed income contracts (as well as complex model specifications in the HJM framework) lead to high-dimensional pricing problems as opposed to pricing typical derivatives in most asset market models.

\section{Market Incompleteness}\label{market incompleteness}
Empirical evidence shows that volatility risk in fixed income markets cannot be hedged by trading solely bonds, which is referred to as \textit{unspanned stochastic volatility} and contradicts many traditional term structure models. By using data on interest rate swaps, caps, and floors, \citet*{collindufresnegoldstein2002} showed that prices of caps and floors, i.e., derivatives exposed to volatility risk, are driven by factors that do not affect prices of interest rate swaps, i.e., the term structure. Therefore, derivatives exposed to volatility risk cannot be replicated by a portfolio consisting solely of bonds, which implies that it is not possible to hedge volatility risk in fixed income markets. The empirical findings of \citet*{collindufresnegoldstein2002} contradict many traditional term structure models, since bond prices are typically functions depending on all risk factors driving the model and bonds can typically be used to hedge caps and floors. As a consequence, \citet*{collindufresnegoldstein2002} examined which term structure models exhibit unspanned stochastic volatility; this led to the development of new models displaying unspanned stochastic volatility \citep*{casassuscollindufresnegoldstein2005,filipoviclarssonstatti2019,filipoviclarssontrolle2017}.
\par In the presence of volatility uncertainty, term structure models naturally exhibit unspanned stochastic volatility, since volatility uncertainty naturally leads to market incompleteness. As mentioned in Remark \ref{remark on hedging}, a classical result in the literature on robust finance is that model uncertainty leads to market incompleteness: instead of perfectly hedging derivatives, one has to superhedge the payoff of most derivatives, which can be inferred from the pricing-hedging duality. Similar to the pricing-hedging duality in the presence of volatility uncertainty \citep*[Theorem 3.6]{vorbrink2014}, we can show that it is not possible to hedge a contract with an asymmetric payoff with a portfolio of bonds. From Theorems \ref{cap price} and \ref{floor price}, we can deduce that caps and floors have an asymmetric payoff if $\overline{\sigma}>\underline{\sigma}$. Therefore, derivatives exposed to volatility risk cannot be hedged by trading solely bonds when there is volatility uncertainty.
\par Moreover, the uncertain volatility affects prices of nonlinear contracts, while prices of linear contracts and the term structure are robust with respect to the volatility---confirming the empirical findings of \citet*{collindufresnegoldstein2002}. In simple model specifications, bond prices have an affine structure with respect to the short rate and an additional factor \citep*[Examples 4.1, 4.2]{holzermann2021'}. They are, however, completely unaffected by the uncertain volatility and its bounds. The same holds for the swap rate, since the price of an interest rate swap (by Proposition \ref{interest rate swap price}) is a linear combination of bond prices, as in the classical case without volatility uncertainty. On the other hand, the uncertain volatility influences prices of caps and floors, since they depend on the bounds for the volatility (by Theorems \ref{cap price} and \ref{floor price}). Therefore, the prices of caps and floors are driven by an additional factor that does not influence term structure movements and (thus) changes in swap rates.

\section{Conclusion}\label{conclusion}
In the present paper, we deal with the pricing of contracts in fixed income markets under Knightian uncertainty about the volatility. The starting point is an arbitrage-free HJM model with volatility uncertainty. Such a framework leads to a sublinear pricing measure, which yields either the price of a contract or its pricing bounds. We derive various methods to price all major interest rate derivatives. We find that there is a single price for typical linear contracts, which is the same as in traditional term structure models; thus, the traditional pricing formulas are completely robust with respect to the volatility. There is a range of prices for typical nonlinear contracts, which is bounded by the prices from the corresponding HJM model without volatility uncertainty for different volatilities. In fact, this applies to all contracts that correspond to (a collection of) convex (or concave) options on forward prices; hence, one can use traditional pricing methods to price such contracts. If the options are not convex (or concave), the prices are characterized by nonlinear PDEs; then one has to rely on numerical schemes. From a theoretical point of view, the main insight is that the pricing formulas are in line with empirical evidence in contrast to traditional pricing formulas.
\par From a practical perspective, the robust pricing procedure developed in this paper provides a theoretical framework for stress testing by pricing contracts in the presence of different levels of volatility uncertainty. When pricing interest rate derivatives in a specific HJM model without volatility uncertainty, one can additionally investigate how robust the prices are with respect to the volatility by allowing for a certain degree of uncertainty about the volatility. For this purpose, one compares the price in an HJM model driven by a standard Brownian motion with the pricing bounds in an HJM model driven by a $G$-Brownian motion with extreme values $\overline{\sigma}=\diag(1+\epsilon,...,1+\epsilon)$ and $\underline{\sigma}=\diag(1-\epsilon,...,1-\epsilon)$ for some $\epsilon>0$. In this way, one can observe how much uncertainty about the price of the contract a certain degree of volatility uncertainty causes.
\par Instead of specifying the level of uncertainty about the volatility, one can also infer it from market data to price other instruments. In order to obtain the extreme values for the volatility, one can fit the range of prices resulting from the pricing procedure in this paper to a spread of prices observed in reality---for example, when prices of contracts are quoted in the form of bid-ask spreads. In this case, the extreme values $\overline{\sigma}$ and $\underline{\sigma}$ are determined in such a way that the upper and the lower bound for the price of a contract match its ask and its bid price, respectively. Then one can use the extracted bounds for the volatility to price other contracts whose prices are not quoted.
\par Alternatively, one can infer the level of uncertainty about the volatility from data on the historical volatility. By looking at historical variations of the volatility, one can extract the bounds for the volatility in the form of a confidence interval in order to obtain a confidence interval for possible prices of a contract. However, this approach is less reasonable than the previous one as the bounds for the volatility represent extreme values for the future evolution of the volatility, which can be very distinct from its past behavior. Prices of options---especially options exposed to volatility risk, such as caps and floors---reflect the market's belief about the future volatility; thus, they provide a better estimate for the future evolution of the volatility (at least from the market's perspective).

\bibliography{C:/Users/jhoelzermann/Documents/Uni/Literature/Literature}
\bibliographystyle{chicago}

\end{document}